\documentclass[12pt]{article}%

\usepackage{amssymb}
\usepackage{amsfonts}
\usepackage{amsmath}
\usepackage[nohead]{geometry}
\usepackage[singlespacing]{setspace}
\usepackage[bottom]{footmisc}
\usepackage{indentfirst}
\usepackage{endnotes}
\usepackage{graphicx}%
\usepackage{rotating}
\usepackage{verbatim}
\usepackage{setspace}
\usepackage{multirow}
\usepackage{latexsym}
\usepackage{subfigure}

\usepackage{amsthm}
\usepackage{makeidx}
\usepackage{fancyhdr}
\usepackage{type1cm}

\newcommand{\bb}{\mathrm{\bf b}}
\newcommand{\bff}{\mathrm{\bf f}}

\newcommand{\by}{\mathrm{\bf y}}

\newcommand{\bA}{\mathrm{\bf A}}

\newcommand{\bu}{\mathrm{\bf u}}

\newcommand{\bB}{\mathrm{\bf B}}

\newcommand{\bC}{\mathrm{\bf C}}
\newcommand{\bD}{\mathrm{\bf D}}
\newcommand{\bE}{\mathrm{\bf E}}
\newcommand{\bF}{\mathrm{\bf F}}

\newcommand{\bH}{\mathrm{\bf H}}
\newcommand{\bI}{\mathrm{\bf I}}
\newcommand{\bV}{\mathrm{\bf V}}

\newcommand{\bR}{\mathrm{\bf R}}
\newcommand{\bS}{\mathrm{\bf S}}

\newcommand{\bone}{\mathrm{\bf 1}}

\newcommand{\bzero}{\mathrm{\bf 0}}

\newcommand{\bxi}{\mbox{\boldmath $\xi$}}
\newcommand{\beps}{\mbox{\boldmath $\varepsilon$}}
\newcommand{\bmu}{\mbox{\boldmath $\mu$}}

\newcommand{\bSigma}{\mbox{\boldmath $\Sigma$}}
\newcommand{\bOmega}{\mbox{\boldmath $\Omega$}}

\newcommand{\bPhi}{\mbox{\boldmath $\Phi$}}

\newcommand{\hB}{\widehat \bB}

\newcommand{\hU}{\widehat U}

\newcommand{\tC}{\widetilde \bC}
\newcommand{\tD}{\widetilde \bD}

\newcommand{\hF}{\widehat \bF}
\newcommand{\hf}{\widehat \bff}

\newcommand{\hu}{\widehat \bu}
\newcommand{\hvar}{\widehat \var}
\newcommand{\hcov}{\widehat \cov}
\newcommand{\hgamma}{\widehat \gamma}

\newcommand{\hSig}{\widehat\Sig}

\newcommand{\hsig}{\widehat\sigma}
\newcommand{\hlam}{\widehat\lambda}

\newcommand{\hxi}{\widehat\bxi}

\newcommand{\cov}{\mathrm{cov}}
\newcommand{\Sig}{\mathbf{\Sigma}}

\newcommand{\tr}{\mathrm{tr}}

\newcommand{\argmin}{\mathrm{argmin}}

\newcommand{\bw}{\mbox{\bf w}}

\newcommand{\var}{\mathrm{var}}
\newcommand{\beq}{\begin{eqnarray*}}
\newcommand{\eeq}{\end{eqnarray*}}

\numberwithin{equation}{section}
\theoremstyle{plain}
\newtheorem{thm}{Theorem}[section]

\newtheorem{lem}{Lemma}[section]
\newtheorem{cor}{Corollary}[section]

\newtheorem{assum}{Assumption}[section]
\theoremstyle{definition}
\newtheorem{exm}{Example}[section]
\newtheorem{remark}{Remark}[section]

\makeatletter
\def\@biblabel#1{\hspace*{-\labelsep}}
\makeatother
\begin{document}

\title{ 
\bf Risks of Large Portfolios}
\author{Jianqing Fan
\thanks{Address: Department of Operations Research and Financial Engineering, Sherrerd Hall, Princeton University, Princeton, NJ 08544, USA, e-mail: \textit{jqfan@princeton.edu},
\textit{yuanliao@umd.edu}, \textit{xshi@princeton.edu}.  The research was partially supported by DMS-1206464 and NIH R01GM100474-03,
NIH R01-GM072611-08.}$\; ^\dag$, Yuan Liao$^\ddag$ and Xiaofeng Shi$^*$
\medskip\\{\normalsize $^*$Department of Operations Research and Financial Engineering,  Princeton University}
\medskip\\{\normalsize $^\dag$ Bendheim Center for Finance, Princeton University}
\medskip\\{\normalsize $^\ddag$ Department of Mathematics, University of Maryland}}

\date{}
\maketitle

\sloppy%

\onehalfspacing


\begin{abstract}
Estimating and assessing the risk of a large portfolio is an important topic in financial econometrics and risk management. The risk is often estimated by a substitution of a good estimator of the volatility matrix. However, the accuracy of such a risk estimator for large portfolios is largely unknown, and a simple inequality in the previous literature gives an infeasible upper bound for the estimation error. In addition, numerical studies illustrate that this upper bound is very crude.  In this paper, we propose  factor-based risk estimators under a large amount of assets, and introduce a  high-confidence level upper bound (H-CLUB)  to assess the accuracy of the risk estimation. The H-CLUB is constructed based on  three different estimates of the volatility matrix: sample covariance, approximate factor model with known factors, and unknown factors (POET, Fan, Liao and Mincheva, 2013). For the first time in the literature,  we derive the limiting distribution of the estimated risks in high dimensionality. Our numerical results demonstrate that the proposed upper bounds significantly outperform the traditional crude bounds, and provide insightful assessment of the estimation of the  portfolio risks. In addition,  our simulated results quantify the relative error in the risk estimation, which is usually negligible using 3-month daily data. Finally, the proposed methods are applied to an empirical study.

\end{abstract}
\strut

\strut

\pagebreak%
\doublespacing

\onehalfspacing

\section{Introduction}

The potential of a portfolio's loss is termed as the portfolio risk. There are two types of portfolio risks. The {\it systematic risk} (or {\it market risk}) is the risk inherent to the entire market, such as risk associated with interest rates, currencies, recession, war and political instability, etc. The systematic risk cannot be diversified away, even with a well-diversified portfolio. In contrast, {\it specific risk} (or {\it idiosyncratic risk}) refers to the risk that affects a very specific group of securities or even an individual security. For example, it can be the risk of price changes due to the unique circumstances of a specific stock. Unlike systematic risk, specific risk can be reduced through diversification.

Estimating and assessing the risk of a large portfolio is an important topic in financial econometrics and risk management. The risk of a given portfolio allocation vector $\bw$ is conveniently measured by $(\bw'\Sig\bw)^{1/2}$, in which $\Sig$ is a  volatility (covariance) matrix of the assets' returns.
Often multiple portfolio risks are at interests and hence it is essential to estimate the volatility matrix $\Sig$.   The problem becomes   challenging when the portfolio size is large. Suppose we have created a portfolio from two thousand assets and invested in a part of selected assets. The covariance matrix $\Sig$ involved then  contains over two million unknown parameters.  Yet, the sample size based on one year's daily data is around 252. It is hard to assess the estimation accuracy when the estimation errors from more than two million parameters are aggregated. Hence some regularization method is recommended to estimate and assess risks.

The interest on large portfolios surges recently.  Pesaran and Zaffaroni (2008) examined the asymptotic behavior of the portfolio weights.    Brodie et al.  (2009)  addressed the problem of portfolio selection using a regularization penalty.  Gomez and Gallon (2011) numerically compared several methods of covariance matrix estimation for portfolio management.  In particular, the optimal portfolio selection involves inverting an estimated  $\Sig$, which is a challenging problem  under a large number of assets. The literature is also found in Jacquier and Polson (2010), Antoine (2011),  Chang and Tsay (2010), DeMiguel et al. (2009),  Ledoit and Wolf (2003), El Karoui (2010),  Lai et al. (2011),  Bannouh et al. (2012), Gandy and Veraart (2012), Bianchi and Carvalho (2011), among others.  


 This paper contributes to the literature   in at least four aspects. First of all,  we propose  risk estimators based on factor analysis.  Traditionally $\Sig$ is estimated by the sample covariance. However, when the number of assets is larger than the sample size, it is well known that the sample covariance is singular, which may result in an estimated risk being zero for certain portfolios. 
 By assuming a factor structure on the returns,  we obtain strictly positive definite covariance estimators $\hSig$ even when the number of assets is larger than the sample size. Two factor-based methods are proposed. The first estimator assumes the factors to be known and observable. The second method deals with the case in which common factors are unknown. This is particularly important for analyzing many non-U.S. markets when assets' returns are driven by a few unknown factors.   In both cases,  the factor model imposes a \textit{conditionally sparse} structure, in that the idiosyncratic covariance is a large sparse matrix. This yields to an \textit{approximate factor model} as in Chamberlain and Rothschild (1983), with a non-diagonal error covariance matrix. 


 Secondly, we provide a new and practical method to assess the accuracy of risk estimation $\bw'(\hSig-\Sig)\bw$. In the literature (e.g., Fan et al. 2012), this term has been bounded by $$
\xi_T=\|\bw\|_1^2\|\hSig-\Sig\|_{\max}
$$
 where $\|\bw\|_1$, the $L_1$-norm of $\bw$,  is the gross exposure of the portfolio, which is bounded when there are no extreme positions in the portfolio.  However, this upper bound depends on the unknown $\Sig$, hence is not applicable in practice. In addition, the numerical studies in this paper demonstrate that this upper bound is too crude:  it is often of the same or even larger scale than the estimated risk.  In contrast, we provide a  high-confidence level upper bound (H-CLUB) for $\bw'(\hSig-\Sig)\bw$, which is of much smaller scale and  easy to compute  in practice.  H-CLUB is constructed based on the confidence interval for the true risk. For each proposed risk estimator $\bw'\hSig\bw$ and a given $\tau\in(0,1)$, we find an H-CLUB  $\hU(\tau)$ such that $$P(|\bw'(\hSig-\Sig)\bw|\leq \hU(\tau))\rightarrow1-\tau.$$ In contrast, $P(|\bw'(\hSig-\Sig)\bw|\leq\xi_T)=1.$
  Hence H-CLUB is an upper bound for the risk estimation error with high confidence while the traditional bound $\xi_T$ is of full confidence. 

  The third contribution is that   for the first time in the literature, we derive the inferential theory of the risk estimators with a high-dimensional portfolio, especially when the estimator is factor-based  with either observed or unobserved factors.  Although factor analysis has long been used  for the portfolio allocation theory, it remains largely unknown whether the effects of estimating the factor loadings and unobservable factors are  negligible in risk estimation, especially when the dimensionality is high. This paper proves that these effects are indeed asymptotically negligible for diversified portfolios, even when the dimensionality is much larger than the sample size.    Interestingly, we find that when the dimensionality is larger than the sample size, the factor-based risk estimators have the same asymptotic variances no matter whether the factors are known or not, and   they are asymptotically equivalent. Hence  the high dimensionality is in fact a bless for risk estimation instead of a curse from this point of view. In addition, the asymptotic variance of factor-based estimators is  slightly  smaller than that of the sample covariance-based estimator, but the difference is  small. This demonstrates that the benefit of using a factor model is not in terms of a much smaller asymptotic variance, because the systematic risk cannot be diversified. Rather, factor analysis gives a strictly positive definite covariance estimator, which is essential  to estimate the optimal portfolio allocation vector, and also interprets the structure of the portfolio risks. 


Finally, using our simulated results based on the model calibrated from the real U.S. equity market data, we are able to  quantify the relative error of the estimation error or coefficient of variation, defined as STD$(\bw'\hSig\bw)/\bw'\hSig\bw$, where STD$(\cdot)$ denotes the standard error of the estimated risk. Interestingly, this ratio is just a few percent and is approximately independent of the gross exposure $\|\bw\|_1$ but sensitive to the length of the time series. On the other hand, we also quantify the relation between the crude bound and the practical H-CLUB. We find that $\xi_T$ is many times larger than $\hU(\tau)$, and  the ratio $\xi_T/\hU(\tau)$   increases as the gross exposure  increases.

We also contribute to the portfolio theory by introducing a sampling technique which picks a random portfolio with a given  gross exposure level.  This sampling scheme can be useful for portfolio optimization and understanding the overall risks within a given level of gross exposure.

We emphasize that the recent works by Fan et al. (2011, 2013) are only concerned about covariance estimations and no inferential theories  were studied. In contrast, we focus on the risk estimation, with a particular attention to the risk assessment and the impact of covariance estimation on the limiting distributions of risk estimators.

The rest of the paper is organized as follows. Section 2 proposes new risk estimators based on   factor analysis under both known and unknown factor cases. Section 3 constructs the H-CLUB for each risk estimator based on the confidence interval for risks. It also derives the limiting distributions of the risk estimators and compares their asymptotic variances. Section \ref{simulation} presents simulation results. An empirical study is considered in Section 5. Finally, Section 6 concludes. All the proofs are given in the appendix.

 Throughout the paper, $\|\bw\|_1=\sum_{i=1}^N|w_i|$ is used to denote the gross exposure of a given portfolio allocation vector.  For a square matrix $\bA$, $\lambda_{\min}(\bA)$ and $\lambda_{\max}(\bA)$ represent its minimum and maximum eigenvalues. Let  $\|\bA\|_{\max}$ and $\|\bA\|$  denote its  element-wise sup-norm and operator norm, given by $\|\bA\|_{\max}=\max_{i,j}|A_{ij}|$  and $\|\bA\|=\lambda_{\max}^{1/2}(\bA'\bA)$ respectively.
\section{Estimation of Portfolio's Risks}

Let$\{\bR_t\}_{t=1}^T$ be a time series of an $N\times 1$ vector of observed asset returns and $\Sig=\cov(\bR_t)$,  often known as the volatility matrix.  The portfolio risk of a given allocation vector $\bw$ is given by $\sqrt{\var(\bw'\bR_t)}$, which is $\sqrt{\bw'\Sig\bw}$. How to estimate the risk of a large portfolio? A straightforward answer is $\sqrt{\bw'\hSig\bw}$  with an estimator $\hSig$. But, how good is it? How to assess the accuracy of this estimator?  We address the problem or risk estimation in this section. The assessment of the estimation accuracy will be discussed in Section 3.

The problem of estimating  the risk of a given portfolio  is challenging   due to the high dimensionality of $\Sig$. Often the number of assets can be of hundreds or even thousands. On the other hand, to adapt to the current market condition,  a short period of financial data    are often  used. For example,  the number of daily returns in three months is only of tens. Hence $N$ can be much larger than $T$.   We assume $\Sig$ to be time-invariant within a short period, which holds approximately for locally stationary  time series.

We consider three estimators  for estimating $\var(\bw'\bR_t)$ for a given $\bw$, based on three different estimators $\hSig$: sample covariance estimator, factor analysis with either observed or unobserved factors. Recently, Chang and Tsay (2010) proposed a Cholesky decomposition approach to estimate the large covariance matrix, and used simulation to assess its performance. On the other hand, the assets' returns are usually driven by a few market factors. Due to the presence of these common factors, $\Sig$ itself is not sparse. Moreover, as pointed out by Stock and Watson (2002) and Bai (2003), common factors are usually pervasive, so the factor loading matrix is not sparse either. Hence the factor-based risk estimators are widely applicable in analyzing financial data, whose asymptotic properties (as both $T, N\rightarrow\infty$) will be also presented below.

 \subsection{Sample-covariance-based estimator}

The first estimator   $\hSig=\bS$ is the conventional sample covariance matrix based on $\{\bR_t\}_{t=1}^T$. The asymptotic impact of using $\bS$ on the risk management has been studied by Fan et al. (2008, 2012) when $N$ is much larger than $T$. The sample covariance estimator does not require any structural assumption on the assets' returns.  It was shown by the aforementioned authors that for a given portfolio $\bw$ with a bounded gross exposure (that is, $\|\bw\|_1$ is bounded),
$$\bw'(\bS-\Sig)\bw=O_p(\sqrt{\frac{\log N}{T}}).$$
However, when $N>T$, it is well known that $\bS$ is singular, and therefore may result in an estimated risk being zero for certain portfolios.

\subsection{Estimating  risks based on factor analysis}
To overcome the problem of singularity of the sample covariance under high dimensionality, we assume that $\bR_t$ satisfies an ``approximate factor model" (Chamberlain and Rothschild 1983):
\begin{equation}
\bR_{t}=\bB \bff_t+\bu_t, t\leq T,
\end{equation}
where $\bB$ is an $N\times K$ matrix of factor loadings; $\bff_t$ is a $K\times 1$ vector of common factors, and $\bu_t$ is an $N\times 1$ vector of idiosyncratic error components.  In contrast to $N$ and $T$,  here $K$ is assumed to be fixed. The common factors may or may not be observable. For example, Fama and French (1992, 1993) identified three known factors that have successfully described the U.S. stock market. In addition, macroeconomic and financial market variables have been thought to capture systematic risks as observable factors. On the other hand, in an empirical study, Bai and Ng (2002) determined two unobservable factors for stocks traded on the New York Stock Exchange during 1994-1998.

Let $\cov(\bff_t)$ and $\Sig_u=\cov(\bu_t)$ denote the covariance matrices of $\bff_t$ and $\bu_t$, $K\times K$ and $N\times N$ respectively. Suppose  $\bff_t$ and $\bu_t$ are uncorrelated. The factor model then implies the following   decomposition of $\Sig$:
$$
\Sig=\bB  \cov(\bff_t) \bB'+\Sig_u.
$$
Sparsity is one of the most common  structures for large covariance estimation,  which assumes many off-diagonal elements of the covariance to be either zero or nearly so. 
 In the approximate factor model, a natural  assumption is to place a sparse structure on $\Sig_u$.  The rationale is, after the common factors are taken out, the remaining idiosyncratic components should be mostly weakly correlated with each other. Such a condition is called \textit{conditional sparsity}. We now propose new risk estimators based on the conditional sparsity assumption.

\subsubsection{Factor-based estimator}

We first  assume that the common factors are observable, and construct an estimator of $\Sig$ based on thresholding on the covariance matrix of idiosyncratic errors.   Suppose $\hB$ is the least squares estimator of $\bB$. The residual sample covariance matrix of $\bu_t$ is then given by $$\bS_u=T^{-1}\sum_{t=1}^T\hu_t\hu_t'=(S_{u,ij})_{N\times N}, \quad \hu_t=\bR_t-\hB\bff_t.$$
 Let $s_{ij}(\cdot):\mathbb{R}\rightarrow\mathbb{R}$ be an entry-dependent adaptive thresholding function and for some thresholding parameter  $\tau_{ij}^f>0$,
 \begin{equation}\label{e2.2}
s_{ij}(z)=0 \text{ when } |z|\leq \tau_{ij}^f, \quad \text{ and } |s_{ij}(z)-z|\leq \tau_{ij}^f.
\end{equation}
A simple example is that $s_{ij}(z) = z I(|z| \geq \tau_{ij}^f)$ with $\tau_{ij}^f = \tau (s_{u, ii}
s_{u, jj})^{1/2}$, namely, setting all correlation coefficients smaller than $\tau$ to zero.  This rule is called hard thresholding in the literature.  The soft-thresholding rule is given by
$s_{ij}(z) = (z - \tau_{ij}^f)_+$.
Let
$$
\widehat{\Sigma}_{u,ij}=\begin{cases}
S_{u,ii}, & i=j\\
s_{ij}(S_{u,ij}), & i\neq j.
\end{cases}
$$
Let $\hcov(\bff_t)$ denote the sample covariance of the common factors.
Define the estimated covariance matrices as
 \begin{equation}\label{eq2.3}
  \hSig_f=\hB\hcov(\bff_t)\hB'+\hSig_u, \quad \hSig_u=(\widehat{\Sigma}_{u,ij})_{N\times N}.
 \end{equation}
 The first condition of  the thresholding function plays a role of thresholding. When applied to a sample covariance,  it thresholds off most of the small entries that are likely due to the estimation errors. The second condition in (\ref{e2.2}) is used for ``shrinkage", which helps to produce a positive definite covariance estimator for a given finite sample.
  Commonly used examples of $s_{ij}(\cdot)$ include hard-thresholding, soft-thresholding, SCAD thresholding, etc. See Antoniadis and Fan (2001), Rothman et al. (2009) and Cai and Liu (2011) for details.

The cut-off is taken to be, for some $C>0$,
$$
\tau_{ij}^f=C\sqrt{S_{u,ii}S_{u,jj}}\sqrt{\frac{\log N}{T}},
$$
which corresponds to applying the thresholding with parameter $C\sqrt{\log N/T}$ to the correlation matrix of $\bS_u$. One can adjust $C$ to gain the strictly positive definiteness of $\hSig_f$ for any given finite sample (see the discussion in Fryzlewicz 2012).

\subsubsection{POET estimator}

When the common factors are unobservable, we estimated $\Sig$ by ``principal  orthogonal complements thresholding" (POET), recently proposed by Fan et al. (2013). The POET works as follows: Let $\hlam_1\geq \cdots \geq\hlam_N$ be the ordered eigenvalues of the sample covariance $\bS$, whose corresponding eigenvectors are denoted by $\{\hxi_j\}_{j=1}^N$.   We then estimate $\Sig$ by
$$
\hSig_{P}=\sum_{j=1}^K\hlam_j\hxi_j\hxi_j'+ \bOmega,\quad \bOmega=(\Omega_{ij})_{N\times N},   \quad \Omega_{ij}=\begin{cases}
\sum_{k=K+1}^N\hlam_k\widehat{\xi}_{k,i}^2, & i=j\\
s_{ij}(  \sum_{k=K+1}^N\hlam_k\widehat{\xi}_{k, i} \widehat{\xi}_{k, j}  ), & i\neq j.
\end{cases}
$$
 where $s_{ij}(\cdot)$ is the same  adaptive thresholding function as before, based on an entry-dependent threshold $\tau_{ij}^P$:
$$
\tau_{ij}^P=C\sqrt{\Omega_{ii}\Omega_{jj}}\left(\sqrt{\frac{\log N}{T}}+\frac{1}{\sqrt{N}}\right).
$$
Recall that $K$ denotes the number of common factors.  Here $C$ is a user-specified constant to maintain the finite sample positive definiteness.  Thanks to the thresholding,  even when $T=o(N)$,   there is $C^*>0$ such that for any $C>C^*$,  both $\hSig_f$ and $\hSig_P$ are strictly positive definite with probability approaching to one.
Simulated and empirical studies suggested that $C=0.5$ is a good choice when $s_{ij}$ is the soft thresholding.

Based on the factor analysis, our  proposed risk estimator is either  $
\sqrt{\bw'\hSig_f\bw}$ or $ \sqrt{\bw'\hSig_P\bw}
$ for a given portfolio allocation vector $\bw$, depending on whether $\bff_t$ is observable. Note that Fan et al. (2013)  is concerned only about the covariance estimation. In contrast,  this  paper focuses on the asymptotic behaviors of these risk estimators and their assessment for a given diversified $\bw$, which  have never been addressed before.   We will see that under high dimensionality, the factor-based estimators have the same asymptotic variance, and is smaller than that of the sample covariance-based estimator. The effect of estimating the unknown factors on the limiting distributions is asymptotically negligible.

\section{Assessment of the Risk Estimation}
This section proposes a new method to assess the estimated risks for a given portfolio allocation vector $\bw$. We will assume $\|\bw\|_1\leq c$ for some $c>0$, where $\|\bw\|_1$ is the gross exposure of the portfolio.  This prevents extreme positions.
\subsection{Measuring risks using full confidence bound}

As described in Section 2, we use a covariance estimator to form a risk estimator    $\bw'\hSig\bw$. 
A natural question then arises: how close is the estimated risk to the true risk? In other words, how do we assess $\Delta=|\bw'(\hSig-\Sig)\bw|$? Technically this question is challenging under high dimensionality.  A simple inequality as $\Delta\leq \|\bw\|^2\|\hSig-\Sig\|$ would not give a convergence upper bound when $N$ is large. 
An alternative (and commonly used) upper bound for $\Delta$ is based on the following inequality:
\begin{equation} \label{eq2.2}
\Delta\leq\|\bw\|_1^2\|\hSig-\Sig\|_{\max}\equiv \xi_T,
\end{equation}
which is usually tighter for risk assessment.

However, for the purposes of  statistical inference, $\xi_T$ is infeasible as it depends on the true $\Sig$. As a result, this upper bound cannot be evaluated in practice for a given data set. In addition, our simulation results  have shown that the upper bound $\xi_T$ is actually too crude to be useful.  Let us consider the following toy example.
\begin{exm}
Consider three stocks with annualized returns that jointly follow a multivariate Gaussian distribution $N_3(0,\bSigma)$ where $\bSigma=0.04\cdot \mathbf{I}_3$. An equally weighted portfolio $\bw=(1/3,1/3,1/3)'$ is constructed and the task is to estimate the portfolio risk using the sample covariance matrix $\bS$ based on the simulated 21-day (one month) returns.

The theoretical value of portfolio variance is  $\bw'\bSigma\bw=0.0133$, which corresponds to a true risk of 11.55\% per annum. Based on a typical simulated data, the estimated portfolio variance $\bw\hSig\bw=0.0131$, equivalent to a perceived risk of 11.43\% per annum. Moreover $\xi_T = \|\bSigma-\hSig\|_{\max} = 0.0248$.  Based on this upper bound, a simple calculation shows that $\bw'\bSigma\bw\in[0, 0.0379]$, that is, the true  risk $\sqrt{\bw'\Sig\bw}$ lies in $[0, 19.46\%]$, an interval that is too wide to be meaningful.
 $\square$
\end{exm}

Note that the inequality (\ref{eq2.2}) holds for every sampling sequence $\{\bR_t\}_{t=1}^T$. Hence $\xi_T$ is in fact an upper bound of full confidence, that is,
$$
P(|\bw'(\hSig-\Sig)\bw|\leq \xi_T)=1.
$$
The toy example is typical in the sense that $\xi_T$ is already too crude for small portfolios. In statistical inference, often people use bounds of high confidence levels instead, e.g., quantities that bound $\Delta$ with a high probability. This paper  pursues such a high-confidence-level upper bound (H-CLUB) based on the confidence interval.

\subsection{H-CLUB}

We propose a new confidence upper bound for $\Delta=|\bw'(\hSig-\Sig)\bw|$ to assess the estimation error of the portfolio risks. More specifically, for each proposed matrix estimator $\hSig$ and  any given $\tau>0$, we find a quantity $\hU(\tau)$  such that  for all large $N$ and $T$,
$$
P(|\bw'(\hSig-\Sig)\bw|\leq \hU(\tau))\geq1-\tau.
$$
  Therefore, $\hU(\tau)$ is an asymptotic $(1-\tau)100\%$ confidence upper bound for $\Delta$. In addition, it is data-driven (up to  user-specified tuning parameters), hence can be easily  calculated in practice and used to construct confidence intervals for the true risks.


Before proceeding, we make a technical comment that one needs to be careful about the limiting behaviors of $T$ and $N$. In this paper, we will treat $N$  as an increasing function of $T$. Hence $N$ grows via a fixed trajectory, e.g., $N=N_T=T^{\alpha}$ for some $\alpha>0$, and can be faster than   $T$, namely, $\alpha>1$.   As a result, we need to apply the triangular array central limit theorem with weakly dependent time series data.

\subsection{Sample covariance based risk estimator}

Let us start with  the sample covariance matrix of $\bR_t$. For simplicity and exposition, let us assume that the returns have mean zero and $\bS=T^{-1}\sum_{t=1}^T\bR_t\bR_t'.$ We make the following assumptions, under which the serial dependence  across $t$ is allowed.

\begin{assum} \label{a31}
(i) $\{\bR_t\}_{t=1}^T$  is strictly stationary with $E\bR_t=0$ and $\cov(\bR_t)=\Sig$.\\
(ii)  There is $M>0$ such that $\max_{i\leq N}E|R_{it}|^8\leq M.$
\end{assum}

Let us introduce the strong mixing condition. Let $\mathcal{F}_{-\infty}^0(R)$ and $\mathcal{F}_{T}^{\infty}(R)$ denote the $\sigma$-algebras generated by $\{\bR_t: -\infty\leq t\leq 0\}$ and  $\{\bR_t: T\leq t\leq \infty\}$ respectively. In addition, define the mixing coefficient $\alpha_R(T)=\sup_{A\in\mathcal{F}_{-\infty}^0(R), B\in\mathcal{F}_{T}^{\infty}(R)}|P(A)P(B)-P(AB)|.$


Define the autoregressive function $\gamma_T(h)=\cov((\bw'\bR_t)^2, (\bw'\bR_{t+h})^2)$, which depends on $T$ through $\dim(\bw)=N=N_T.$ 
Let
\begin{equation}\label{e3.2}
  \sigma^2_T   =\gamma_T(0)+2\sum_{h=1}^{\infty}\gamma_T(h).
  \end{equation}

 \begin{assum} \label{a32new} (i) There exists   $r_0>0$ and $M>0$ such that:  for all $T\in\mathbb{Z}^+$,
$$\alpha_R(T)\leq \exp(-MT^{r_0}).$$
 (ii) $\sum_{h=1}^{\infty}|\gamma_T(h)|=O(1)$,  $ \sum_{h=1}^{T}h\gamma_T(h)/T=o(  \sigma^2_T   )$ and $\alpha_R(T)=o(\gamma_T(0))$.
 \end{assum}
 Assumption \ref{a32new} requires the weak dependence of the time series. Strong mixing condition is assumed. The first two conditions in (ii) are usually mild. When a diversified $\bw$ is used, the last condition in (ii) is easy to satisfy as long as the dimensionality is not  exponentially large in $T$ because the mixing coefficient is assumed to decay exponentially fast.     To illustrate its meaning, consider a simple case where $\bw=(1/N,\cdots,1/N)$,  then $\gamma_T(0)=\var((\frac{1}{N}\sum_{i=1}^NR_{it})^2)$, which is in general no smaller than $O(N^{-c})$ for some $c>0$. Due to the strong mixing condition in Assumption \ref{a32new}(i),  $\alpha_R(T)=o(\gamma_T(0))$  if $N=N_T$ grows at a polynomial rate of $T$.

 We are now ready to define the H-CLUB for the estimation error $\bw'(\bS-\Sig)\bw$. Let
 $$
    \hgamma(h)=T^{-1}\sum_{t=1}^{T-h}((\bw'\bR_t)^2-\bw'\bS\bw)((\bw'\bR_{t+h})^2-\bw'\bS\bw).
 $$
  In particular, $\hgamma(0)=T^{-1}\sum_{t=1}^T(\bw'\bR_t)^4-(\bw'\bS\bw)^2.$  Let $z_{\tau/2}$ denote the upper $\tau/2$ quantile of the standard normal distribution.
 For some increasing sequence $L=L(T)\rightarrow\infty$, let
\begin{equation}\label{e32}\hsig^2=\hgamma(0)+2\sum_{h=1}^L\hgamma(h), \quad \hU_S(\tau)=z_{\tau/2}\sqrt{\hsig^2/T}.
\end{equation}
Here $L$ is a truncation parameter, and as $L$ slowly increases, $\hsig^2$ consistently estimates $\sigma_T^2.$

\begin{lem}  \label{l31}
  Under Assumptions  \ref{a31}-\ref{a32new},
$$
|\hsig^2-  \sigma^2_T   |=O_p(L^{3/2}T^{-1/2}+\sum_{h>L}\gamma_T(h)),
$$
If in addition $L^{3}=o(T\sigma^4_T)$ and $\sum_{h>L}\gamma(h)=o(  \sigma^2_T   )$, then
$$
|\hsig^2-  \sigma^2_T   |=o_p(\sigma_T^2) \quad \mbox{and} \quad
\hU_S(\tau)=o\left(\sqrt{\frac{\log N}{T}}\right).
$$
\end{lem}

The following theorem gives the limiting distribution of the estimated risk. It also demonstrates that $\hU_S(\tau)$ is a valid H-CLUB for $|\bw'(\bS-\hSig)\bw|$.

\begin{thm}\label{t31}
 Under the assumptions of Lemma \ref{l31},  as $T, N\rightarrow\infty$,
 $$
\left[\var\left(\sum_{t=1}^T(\bw'\bR_t)^2\right)\right]^{-1/2}   T\bw'(\bS-\Sig)\bw\rightarrow^d\mathcal{N}(0,1),
$$
and for any $\tau>0,$
 $$
 P\left(|\bw'(\bS-\Sig)\bw|\leq  \hU_S(\tau)\right)\rightarrow1-\tau.
 $$
 \end{thm}

As a result  $\hU(\tau)=z_{\tau/2}\sqrt{\hsig^2/T}$ is an H-CLUB with confidence level $(1-\tau)100\%$ and is   data-driven once a user-specified $L$ is determined.  Compared to the traditional bound $\xi_T$, $\hU_S(\tau)$ can be easily calculated for any given time series data.  The scale of  $\hU_S(\tau)$ is  much smaller  than that of $\xi_T.$ Our simulation results show that even for a small $\tau$ (e.g., $\tau=0.01$), the magnitude of $\hU(\tau)$ is much smaller than the crude bound $\xi_T$. (See Table \ref{tab:3} in Section  \ref{simulation}.)

By the $\delta$-method, we have the following corollary for the risk estimation.
Define $$\hat R(\bw)=\sqrt{\bw'\bS\bw},\quad R(\bw)=\sqrt{\bw'\Sig\bw}.$$
\begin{cor}\label{co31}
Under  the assumptions of Lemma \ref{l31},  for any $\tau>0,$ as $T, N\rightarrow\infty$,
$$P\left(|\hat R(\bw)-R(\bw)|\leq \hU_S(\tau)/\sqrt{4\bw'\bS\bw}\right)\rightarrow1-\tau.$$
\end{cor}

\subsection{Factor-based risk estimator}

Let us now approach the problem via factor analysis.  We assume
\begin{equation}\label{e3.1}
\bR_t=\bB\bff_t+\bu_t,
\end{equation}
where in this section, $\{\bff_t\}_{t=1}^T$ are observed common factors.    
 In the approximate factor model, the idiosyncratic covariance $\Sig_u$ is non-diagonal. However, the risk component $\bw'(\hSig_u-\Sig_u)\bw$  introduced by the idiosyncratic error can be diversified away by a selected portfolio allocation vector. Hence the estimation error of the risk only comes from the systematic error brought by the common factors.
 Compared to the sample covariance based risk estimator, factor analysis always gives strictly positive risk estimators even  when $N>T$ for any nonzero allocation vector $\bw$.


For the factor-based risk estimation,  a different set of assumptions are needed instead of those in Section 3.3. First of all, the factor model is assumed to be conditionally sparse, in the sense that $\Sig_u$ is a sparse matrix.  We employ the approximate sparsity assumption in Bickel and Levina (2008) as follows:

\begin{assum}\label{ass3.5}
 There is $q\in[0,1)$ such that
$$
s_N\equiv\max_{i\leq N}\sum_{j=1}^N|\Sigma_{u,ij}|^q=o(\min\{(T/\log N)^{(1-q)/2}, N^{(1-q)/2}\}).
$$
\end{assum}
When $q=0$ we define $s_N=\max_{i\leq N}\sum_{j=1}^NI(\Sigma_{u,ij}\neq0)$ as the maximum number of nonvanishing elements in each row, and the assumption requires that $s_N = o((T/\log N)^{1/2}, N^{1/2})$.   Assumption \ref{ass3.5}, though slightly stronger than those in Chamberlain and Rothschild (1983),   is quite  meaningful  in practice.  For example,  when the idiosyncratic components represent firms' individual shocks, they are either uncorrelated  or weakly correlated among the firms across different industries, because  the industry specific
components are not pervasive for the whole economy (Connor and Korajczyk 1993).

\begin{assum} \label{a3.5}(i) $\{\bu_t, \bff_t\}_{t=1}^T$ is strictly stationary, $\{\bu_t\}_{t=1}^T$ and $\{\bff_t\}_{t=1}^T$ are independent,  and $Eu_{it}=Ef_{jt}=0$ for all $i, j $.\\
 (ii)  There exist $r_1, r_2>0$ and $b_1, b_2>0$, such that for any $s>0$, $i\leq p$ and $j$,
\begin{equation*}
P(|u_{it}|>s)\leq\exp(-(s/b_1)^{r_1}), \quad P(|f_{jt}|>s)\leq \exp(-(s/b_2)^{r_2}).
\end{equation*}
(iii) There is $C>0$ such that $C^{-1}<\lambda_{\min}(\Sig_u)\leq\lambda_{\max}(\Sig_u)<C$, $\max_{i\leq N}ER_{it}^2<C$, $\|\bB\|_{\max}<C$ and $\lambda_{\min}(\cov(\bff_t))>C^{-1}$.
 \end{assum}

 Let $\mathcal{F}_{-\infty}^0$ and $\mathcal{F}_{T}^{\infty}$ denote the $\sigma$-algebras generated by $\{(\bff_t,\bu_t): -\infty\leq t\leq 0\}$ and  $\{(\bff_t,\bu_t): T\leq t\leq \infty\}$ respectively.  Define the mixing coefficient
$
\alpha_f(T)=\sup_{A\in\mathcal{F}_{-\infty}^0, B\in\mathcal{F}_{T}^{\infty}}|P(A)P(B)-P(AB)|.
$
Let $$\gamma_f(h)=\cov((\bw'\bB\bff_t)^2, (\bw'\bB\bff_{t+h})^2)$$ for $h\geq 0.$ It follows from the $\alpha$-mixing condition that $\sum_{h=1}^{\infty}|\gamma_f(h)|=O(1)$ (see Lemma  \ref{la.7} in the appendix). In addition, define
 \begin{equation}\label{eq3.3}
 \sigma_f^2=\gamma_f(0)+2\sum_{h=1}^{\infty}\gamma_f(h).
 \end{equation}

\begin{assum} \label{a32}
(i) There exists   $r_3>0$  and $M>0$ satisfying:  for all $T\in\mathbb{Z}^+$,
$$\alpha_f(T)\leq \exp(-MT^{r_3}).$$
(ii)  $ \sum_{h=1}^{\infty}h\gamma_f(h)/T=o(\sigma_f^2)$   and $\alpha_f(T)=o(\gamma_f(0))$ as $T, N\rightarrow\infty.$
\end{assum}


\begin{assum} \label{a37}
    $\bw'\Sig_u\bw=o(\sigma_f^4 + \sigma_f T^{-q/2}(\log N)^{-(1-q)/2}s_N^{-1})$.
\end{assum}
Note that $\sigma_f^4=O(1)$. This assumption allows $\sigma_f^4$ to decay as $N$ increases due to diversified allocation vectors.
Recall that $q$ is defined in Assumption \ref{ass3.5}.  Assumption \ref{a37} requires $\|\bw\|=o(1)$, which assumes a diversified portfolio to reduce the idiosyncratic risk. To illustrate the intuition, consider the following simple example.
\begin{exm}\label{ex3.2}
Consider a one-factor model on the asset returns:
$$
R_{it}=b_if_t+u_{it}
$$
where $\var(f_t^2)>0$ and $\Sig_u$ is a diagonal matrix. Hence $q=0$ and $s_N=1$. For simplicity, suppose $\{\bff_t\}_{t=1}^T$ are independent across $t$, and thus  $\sigma_f^2=\gamma_f(0)=(\bw'\bB)^4\var(f_t^2)$. As the eigenvalues of $\Sig_u$ are bounded away from   zero and infinity, Assumption \ref{a37} is equivalent to:
\begin{equation}\label{e3.4}
\bw'\bw=o((\bw'\bB)^8 + (\bw'\bB)^2/\sqrt{\log N}),
\end{equation}
which holds if $\bw$ is ``diversified" enough. For example, the  equal-weight allocation $\bw=(1/N,\cdots,1/N)$ gives $\bw'\Sig_u\bw=O(N^{-1})$.  Writing  $C_N=|N^{-1}\sum_{i=1}^Nb_i|$, then (\ref{e3.4}) holds as long as $N^{-1}\sqrt{\log N}=o(C_N^8).$ This is true since $C_N$ is often bounded away from zero. $\square$
 \end{exm}

To construct H-CLUB, we need to first estimate $\gamma_f(h)$.  For $\hcov(\bff_t)=T^{-1}\sum_{t=1}^T\bff_t\bff_t', $   define
$$
    \hgamma_f(h)=T^{-1}\sum_{t=1}^{T-h} [(\bw'\hB\bff_{t+h})^2-\bw'\hB\hcov(\bff_t)\hB'\bw][(\bw'\hB\bff_{t})^2-\bw'\hB\hcov(\bff_t)\hB'\bw],
$$
where $\hB$ is the least squares estimator of $\bB$. For some $L=L(T)$,  define
\begin{equation}\label{e37}
\hsig_f^2=\hgamma_f(0)+2\sum_{l=1}^L\hgamma_f(h), \quad  \hU_f(\tau)=z_{\tau/2}\sqrt{\hsig_f^2/T}.
\end{equation}
 Let $\beta=3r_1^{-1}+1.5r_2^{-1}+r_3^{-1}$.
\begin{lem}\label{l32}
Suppose $ (\log N)^{2\beta+2}=o(T)$ and $L\sqrt{(L+\log N)/T}+\sum_{h>L}\gamma_f(h)=o(  \sigma^2_f   )$. Under Assumptions \ref{ass3.5}-\ref{a37},
 $$
|\hsig_f^2-  \sigma^2_f   |=O_p(L\sqrt{\frac{L+\log N}{T}}+\sum_{h>L}\gamma_f(h)),
$$
and
 $$\hU_f(\tau)=o\left(\sqrt{\frac{\log N}{T}}\right).
 $$
\end{lem}
Hence   the H-CLUB has a smaller stochastic order than that of the crude bound $\xi_T.$

The following theorem shows that $\hU_f(\tau)$ is a valid H-CLUB for the risk estimation error, and can be computed easily from the data in practice.  Technically,  Theorems \ref{t32} and \ref{t33} (to be introduced in the next subsection) below are not simple applications of the triangular array central limit theorem.    We need to show that after thresholding, the idiosyncratic risk can be diversified away by the portfolio vector $\bw$, and the estimation error for the factor loadings  is asymptotically negligible even under high dimensionality.

\begin{thm}\label{t32} Suppose that the common factors are observable, and that the thresholded $\hSig_f$ (\ref{eq2.3}) is used as the covariance estimator. Under the assumptions of Lemma \ref{l32}, as $T, N\rightarrow\infty$,
 $$
\left[\var\left(\sum_{t=1}^T(\bw'\bB\bff_t)^2\right)\right]^{-1/2}   T\bw'(\hSig_f-\Sig)\bw\rightarrow^d\mathcal{N}(0,1),
$$
and for any $\tau>0,$
$$
 P\left(|\bw'(\hSig_f-\Sig)\bw|\leq \hU_f(\tau)\right)\rightarrow1-\tau.
 $$
\end{thm}
 \begin{remark}
 Similar to Corollary \ref{co31}, if we use $\hat R_f(\bw)=\sqrt{\bw'\hSig_f\bw}$ to estimate $R(\bw)=\sqrt{\bw'\Sig\bw}$, then applying a delta method gives
 $$
 P\left(|\hat R_f(\bw)-R(\bw)|\leq \hU_f(\tau)/\sqrt{4\bw'\hSig_f\bw}\right)\rightarrow1-\tau.
 $$Hence $ \hU_f(\tau)/\sqrt{4\bw'\hSig_f\bw}$ is a valid H-CLUB for $|\hat R_f(\bw)-R(\bw)|$.
 \end{remark}

 It is interesting to compare $\hU_f(\tau)$ with $\hU_S(\tau)$ and see if knowing the factor structure results in  a reduced upper bound.    This is equivalent to comparing $\hsig^2$ in (\ref{e32}) with $\hsig_f^2$ in (\ref{e37}). Essentially we are to compare the asymptotic variances of the estimated risks between a pure nonparametric risk estimator (sample covariance) and an estimator based on factor analysis. We will see  in the following section that when the factor structure  is specified, the factor-based risk estimator indeed gives a slightly smaller asymptotic variance.

 \subsection{Risk estimation with unknown factors}

Often the market assets' returns are driven by a few unknown factors. Hence the common factors $\bff_t$ may not be observable which makes the analysis more practical and challenging. In this case, we apply the  POET  estimator for $\Sig$ to handle the difficulty of not knowing the factors:\begin{equation}\label{eq3.8}
\hSig_P=\sum_{j=1}^K\hlam_j\hxi_j\hxi_j'+\bOmega
\end{equation}
with  $K$ being the number of common factors.  For simplicity, we will assume $K$ to be known, and in practice it can be estimated consistently using the BIC method (Bai and Ng 2002). 
Then $K$ in the above estimator can be replaced with its consistent estimator.

Under the conditional sparsity condition, Fan et al. (2011, 2013) showed that
\begin{equation}\label{e3.9}
\|\hSig_P^{-1}-\Sig^{-1}\|=O_p\left(s_N\left(\frac{\log N}{T}+\frac{1}{N}\right)^{1/2-q/2}\right)
\end{equation}
and when common factors are observable,
\begin{equation}\label{e3.10}
\|\hSig_f^{-1}-\Sig^{-1}\|=O_p\left(s_N\left(\frac{\log N}{T}\right)^{1/2-q/2}\right)
\end{equation}
 where $q$ and $s_N$ are defined in Assumption \ref{ass3.5}.
The term $1/N$  in (\ref{e3.9})  is the price for not  knowing $\bff_t$.  When $T=o(N\log N)$, the above convergence rates are the same.  Intuitively, as the dimensionality increases, more information about the common factors is collected, and eventually 
 the common factors can be treated as though they are known.  Moreover,  both  $\hSig_P$ and $\hSig_f$ are strictly positive definite for all large $N$ and $T$.

We will see that with large enough pool of assets and  a diversified portfolio allocation, the effect of estimating the unknown factors on the estimated risk is negligible. As a result, $\bw'\hSig_P\bw$ and $\bw'\hSig_f\bw$ have the same asymptotic limiting distribution.  For this purpose,  we impose  additional  conditions.

\begin{assum}\label{ass3.11}
As $N\rightarrow\infty$, the eigenvalues of $\bB'\bB/N$ are bounded away from both zero and infinity.
\end{assum}

 Intuitively, Assumption \ref{ass3.11}  means that the common factors should be pervasive, that is,  impact on a non-vanishing proportion of individual time series. It implies that the first $K$ eigenvalues of $\Sig$ are growing with rate $O(N)$, which are well separated from  the eigenvalues of $\Sig_u$. 
For identification, we  assume
$
\cov(\bff_t)=\bI_K$ and $\bB'\bB$ to be diagonal.
Consequently, $$\Sig=\bB\bB'+\Sig_u.$$ Write $\bB=(\bb_1,\cdots,\bb_N)'$.  The following assumption is common in the literature of high-dimensional factor analysis, e.g., Bai and Ng (2002), Bai (2003).

\begin{assum}\label{ass3.12} There is $M>0$ such that
 $E[N^{-1/2}(\bu_s'\bu_t-E\bu_s'\bu_t)]^4<M$ and $E\|N^{-1/2}\sum_{i=1}^N\bb_iu_{it}\|^4<M$.
\end{assum}

Motivated by  (\ref{e3.9}) and (\ref{e3.10}), we require $T=o(N\log N)$  so that the effect of estimating the common factors is first-order negligible. This is often true for the asset returns' time series data.  In addition, the portfolio vector $\bw$ should still be diversified enough. This leads to the following assumption:

\begin{assum}\label{a310}
$\sigma_f^2T\rightarrow\infty$, $\sigma_f^2N/T\rightarrow\infty$, and  $\bw'\Sig_u\bw=o(\sqrt{\sigma_f^2}s_N^{-1}N^{1/2-q/2}T^{-1/2}).$
\end{assum}
Assumption \ref{a310} can be verified similarly by an example like Example \ref{ex3.2}.

To define an H-CLUB for a factor model with unknown factors,  we first
apply the principal components method (Stock and Watson 2002 and Bai 2003)   to estimate $\sigma_f^2$.  Let $\hF=(\hf_1,\cdots,\hf_T)$ be a $K\times T$ matrix such that the rows of $\hF/\sqrt{T}$ are the eigenvectors corresponding to the $K$ largest eigenvalues of the $T\times T$ matrix $\bR'\bR$, where $\bR=(\bR_1,\cdots,\bR_T)$. Let $\hB=\bR\hF'/T$. Define
$$
\hgamma_P(h)=T^{-1}\sum_{t=1}^{T-h} [(\bw'\hB\hf_{t+h})^2-\bw'\hB\hB'\bw][(\bw'\hB\hf_{t})^2-\bw'\hB\hB'\bw].
$$
For some $L=L(T)\rightarrow\infty$, let
\begin{equation}\label{e38}
\hsig_P^2=\hgamma_P(0)+2\sum_{h=1}^L\hgamma_P(h), \quad  \hU_P(\tau)=z_{\tau/2}\sqrt{\hsig_P^2/T}.
\end{equation}

\begin{lem}  \label{l33} Suppose $L=o(\sqrt{N}\sigma_f^2)$. Under Assumptions \ref{ass3.5}-\ref{a310},
$$
|\hsig_P^2-  \sigma^2_f   |=O_p(L\sqrt{\frac{L+\log N}{T}}+\frac{L}{\sqrt{N}}+\sum_{h>L}\gamma_f(h))=o_p(\sigma_f^2).
$$
 and
 $$ \hU_P(\tau)=o\left(\sqrt{\frac{\log N}{T}}\right).
 $$
\end{lem}

The following theorem shows that $\hU_P(\tau)$ is an H-CLUB for $\bw'(\hSig_P-\Sig)\bw$, and  can be computed easily from the data.  Interestingly,  $\bw'\hSig_P\bw$ and $\bw'\hSig_f\bw$ have the same asymptotic limiting distribution.  The price  paid for not knowing the   factors is asymptotically negligible. 

\begin{thm}\label{t33} Suppose the common factors are unobservable, and $\hSig_P$ (\ref{eq3.8}) is used as the covariance estimator. Under the assumptions  of Lemma  \ref{l33}, as $T, N\rightarrow\infty$,
 $$
\left[\var\left(\sum_{t=1}^T(\bw'\bB\bff_t)^2\right)\right]^{-1/2}   T\bw'(\hSig_P-\Sig)\bw\rightarrow^d\mathcal{N}(0,1),
$$
and for any $\tau>0,$  $$
 P\left(|\bw'(\hSig_P-\Sig)\bw|\leq \hU_P(\tau)\right)\rightarrow1-\tau.
 $$
\end{thm}
  \begin{remark}
 Similarly,  if we define $\hat R_P(\bw)=\sqrt{\bw'\hSig_P\bw}$, then   $ \hU_P(\tau)/\sqrt{4\bw'\hSig_P\bw}$ is a valid H-CLUB for $|\hat R_P(\bw)-R(\bw)|$.
 \end{remark}

Knowing the factor-structure of the return $\bR_t$ improves the  estimation  efficiency relative to the sample covariance estimator.  This is demonstrated by the following theorem. 

\begin{thm} \label{t44}   Under the assumptions of Theorem \ref{t33},
$$\var\left[\sum_{t=1}^T(\bw'\bR_t)^2\right]> \var\left[\sum_{t=1}^T(\bw'\bB\bff_t)^2\right].$$
\end{thm}

The difference of the above two variances is    actually small when $\bw$ is diversified enough, and this fact is further verified by our simulation results (see Tables \ref{tab:3} and \ref{tab:4} in Section \ref{simulation}). The reason is that the systematic risk cannot be diversified. On the other hand, factor analysis gives a strictly positive definite covariance estimator,  whereas the sample covariance may produce a risk estimator being zero for certain portfolio allocation vectors. The positive definiteness is particularly important  to estimate the optimal portfolio allocation vector. Furthermore, factor analysis interprets the structure of portfolio's risks. It is clearly seen that  the idiosyncratic risks  can be diversified away by the portfolio allocation. 

\section{Monte Carlo Examples}\label{simulation}
In this section, we examine the finite-sample performance of both the full confidence upper bound $\xi_T$ defined in (\ref{eq2.2}) and H-CLUB, based on three covariance estimators $\hSig$, using  portfolios $\bw$ with different gross exposure constraints. Graphical and numerical results illustrate that $\xi_T$ is indeed a very crude bound and H-CLUB has much better performance in general. 
The number of factors and length of time are both fixed with $K=3$ and $T=300$  respectively. The dimensionality $N$ gradually increases from 20 to 600.

Excess returns of the $i$th stock of a portfolio over the risk-free interest rate, $y_{it}$, is assumed to follow the Fama-French three-factor model [Fama and French(1992)]:
$$
y_{it} = \lambda_{i1}f_{1t}+\lambda_{i2}f_{2t}+\lambda_{i3}f_{3t}+u_{it}.
$$
The first factor is the excess return of the whole equity market, while the second and third factors are SMB (``small minus big" cap) and HML (``high minus low" book/price) respectively.  Using  US equity market data, we calibrate a submodel to generate the loadings $\bb_i=(\lambda_{i1},\lambda_{i2}, \lambda_{i3})'$, the idiosyncratic noises $\bu_t$ and the factors $\bff_t=(f_{1t},f_{2t}, f_{3t})'$. 

\subsection{Calibration}
To calibrate parameters in the model, we use the data on daily returns of S\&P 500's top 100 constituents ranked by market capitalization (on June $29^{th}$ 2012),  the data on 3-month Treasury bill rates, and daily return data of the Fama-French factors.  They are obtained from COMPUSTAT database, the data library of Kenneth French's website, and CRSP database respectively. The excess returns ($\tilde{\by}_t, \tilde{\bff}_t$) are analyzed for the period from July $1^{st}$, 2008 to June $29^{th}$ 2012, approximately 1000 trading days.

(1) Calculate the least square estimator $\tilde{\bB}$ of $\tilde{\by}_t = \bB\tilde{\bff}_t+\bu_t$, and compute the sample mean vector $\bmu_B$ and sample covariance matrix $\bSigma_B$ of all the row vectors of $\tilde{\bB}$. These parameters are reported in Table~\ref{tab:B}. The factor loadings $\{\bb_i\}_{i=1}^N$ of the simulated models are then generated from a trivariate Gaussian distribution $\mathcal{N}_3(\bmu_B,\bSigma_B)$.
\begin{table}[ht]
 \centering
 \begin{tabular}{l |llr}
 \hline
 $\hspace{1em}\bmu_B$            &   & $\bSigma_B$ &   \\
\hline
0.9833 &0.0921 & -0.0178 & 0.0436  \\
-0.1233 &-0.0178 & 0.0862 &-0.0211   \\
0.0839&0.0436 & -0.0211 & 0.7624 \\
\hline
\end{tabular}
\caption{Mean and covariance used to generate $\bb_i$}
\label{tab:B}
\end{table}

(2) Assume that the factors follow the stationary vector autoregressive VAR(1) model $\bff_t = \bmu + \bPhi\bff_{t-1}+\beps_t$ for some $3\times 3$ matrix $\bPhi$, where $\beps_t$ follows i.i.d $\mathcal{N}_3(0,\bSigma_{\epsilon})$.  The model parameters $\bPhi,\bmu$ and $\bSigma_{\epsilon}$ are calibrated using the daily excess returns of the Fama-French factors $\tilde{\bff}_t$. The covariance matrix $\cov(\bff_t)$ is then obtained by solving the linear equation $\cov(\bff_t) = \bPhi\cov(\bff_t)\bPhi'+\bSigma_{\epsilon}$. Results are summarized in Table~\ref{tab:f}.
\begin{table}[ht]
 \centering
 \begin{tabular}{l|ccc|ccr}
 \hline
 $\hspace{1em}\bmu$    &   & $\boldsymbol\Phi$ & && $\cov(\bff_t)$&\\
\hline
0.0260 &-0.1006& 0.2803 & -0.0365 & 3.2351&0.1783&0.7783\\
0.0211 &-0.0191 & -0.0944 &0.0186   &0.1783&0.5069&0.0102\\
-0.0043&0.0116& -0.0272 & 0.0272  &0.7783&0.0102&0.6586\\
\hline
\end{tabular}
\caption{Parameters used to generate $\bff_t$}
\label{tab:f}
\end{table}

(3) The error covariance matrix is sparse in our setting. For each fixed $N$, it is created by $\bSigma_u = \bD\bSigma_0\bD$, where $\bD = \text{diag}(\sigma_1,\cdots,\sigma_p)$. To be more specific, $\sigma_1,\cdots,\sigma_p$ are generated independently from a Gamma distribution $G(\alpha,\beta)$, in which $\alpha$ and $\beta$ are  selected to match the sample mean and sample standard deviation of the 100 standard deviations of the   errors $\tilde{\bu}_t=\tilde{\by}_t-\tilde{\bB}\tilde{\bff}_t$ (recall that each $\tilde{\bu}_t$ is 100 dimensional). An additional restriction is imposed on $\sigma_i$   that only values in between the minimum and maximum of the standard deviation of $\tilde{\bu}_t$ are accepted.
We then generate the off-diagonal entries of the correlation matrix $\bSigma_0$ independently from a Gaussian distribution, with mean and standard deviation equal to those of the sample correlations of the estimated residuals.  Moreover, absolute values of the off-diagonal entries are set to no greater than 0.95. Finally the hard-thresholding is applied to make $\bSigma_0$ sparse, where the threshold is set to be the smallest constant that makes $\bSigma_0$ positive definite.

\subsection{Representative portfolios}\label{representative}

We examine the performance of H-CLUB based on $\bw$ with a couple of different gross exposures. For a given exposure $c$ and given number of assets $N$, there are infinitely many portfolios $\bw$ that satisfy
$\sum_{i=1}^N w_i = 1$ and $\sum_{i=1}^N |w_i| = c$.  In order to be representative, we take some portfolios randomly from this set.  This task, which generates uniformly for the above set  in $R^N$, is of independent interest for portfolio optimization and research.  It is also challenging.

Let $w^+$ be the total long position and $w^-$ be the total short position.  Then, $w^+ = (c+1)/2$ and $w^{-} = (c-1)/2$.  For $c = 1$, there are no-short positions.  For $c > 1$, there are both long and short positions.  The identities (or indices) of long and short positions are hard to identify, but the following sampling scheme is a reasonable approximation:  The positive positions are determined by a Bernoulli trial (N times) with probability of success $w^+/(w^+ + w^-) = (c+1)/(2c)$.  Once the identities are determined, we can normalize them and the problem reduces to the case with $c=1$.
For the case with $c = 1$, the uniform distribution on the set $\{w_i:  \sum_{i=1}^N w_i = 1, w_i \geq 0\}$ can be generated from a normalized exponential distribution:
$$
    w_i = \zeta_i/\sum_{i=1}^N \zeta_i, \qquad \zeta_i \sim_{i.i.d.} \mbox{standard exponential}.
$$

Combination of the the above two steps, we can generate a randomly selected portfolio from its feasible set
as follows.
\begin{enumerate}
\item Generate a positive integer $k$, the number of stocks with positive weights in $\bw$, from a binomial distribution $\text{Bin}(N,\frac{c+1}{2c})$.

\item Let $\bw_+ = (w_1^+,\cdots,w_k^+)$ be a temporary vector of the positive weights in $\bw$. Generate independently $\{\zeta_i\}_{i=1}^k$ from the standard exponential distribution and set each $w_i^+ = (c+1)\zeta_i/(2\sum_{j=1}^k\zeta_j)$.

\item The temporary negative weights in $\bw_-=(w^-_1,\cdots,w^-_{N-k})$ are generated analogously with each $w_i^- = (1-c)\zeta_i/2\sum_{j=1}^{N-k}\zeta_j$, where $\{\zeta_j\}_{j=1}^{N-k}$ are obtained independently from the standard exponential distribution.

\item Take the portfolio weights $\bw$ as a random permutation of the vector $(\bw_+,\bw_-)$.
\end{enumerate}

\subsection{Simulation}

For each simulation with a given $c$, we fix  $T=300$ and gradually increase $N$ from 20 to 600 in a multiple of 20. For each fixed $N$, we use 50 different model parameters and 200 testing portfolios for each given set of model parameters so that a total of 10,000 portfolios were actually used.  In other words, we iterate the following steps for $50$ times, record values of $R(\bw), \Delta=|\bw'(\hSig-\Sig)\bw|, \xi_T$ and $\widehat{U}(\tau)$, and compute their own means and standard deviations. The details are summarized as follows:
\begin{enumerate}
\item Generate $\{\bb_i\}_{i=1}^N$ independently from $\mathcal{N}_3(\bmu_B,\bSigma_B)$. Set $\bB=(\bb_1,\cdots,\bb_p)'$.

\item Generate $\{\bu_t\}_{t=1}^T$ independently from $N_p(\bzero,\bSigma_u)$.

\item Generate $\{\bff_t\}_{t=1}^T$ from a VAR(1) model  $\bff_t = \bmu + \bPhi\bff_{t-1}+\beps_t$ with parameters specified in the calibration part.

\item Calculate $\by_t = \bB\bff_t+\bu_t$ for $t=1,\cdots, T$.

\item Calculate the sample covariance matrix $\bS = T^{-1}\sum_{t=1}^T(\by_t-\bar{\by})(\by_t-\bar{\by})'$;  obtain the factor-based covariance estimator $\hSig_f$ by using the hard-thresholding rule with the threshold $\omega_T = 0.10K\sqrt{\log N/T}$;  and get the POET covariance estimator $\hSig_P$ using the soft-thresholding with thresholding parameter $0.5\sqrt{\Omega_{ii}\Omega_{jj}}(\sqrt{\frac{\log N}{T}}+\frac{1}{\sqrt{N}})$. 

\item Generate 200 $\bw$ according to the method described in Section~\ref{representative}.

\item Over the 200 generated portfolios $\bw$, compute the average of true risk $R(\bw) = \sqrt{\bw'\bSigma\bw}$; Also for $\hSig=\bS, \hSig_f$ and $\hSig_P$, compute their respective average of $\Delta = |\bw'(\hSig-\bSigma)\bw|$, $\xi_T = \|\bw\|_1^2\|\hSig-\bSigma\|_{\max}$ and $\widehat{U}(0.05) = 2\sqrt{\hsig^2/T}$. 
In our setting, the number of lags $L=5$.
\end{enumerate}

Under several gross exposure constraints $c$, we produce the graph of risk domain by plotting $R(\bw)$ as a function of $c$ and $N$ (20 to 600 in increments of 20).  Averages of $\Delta$, $\xi_T$ and $\widehat{U}(0.05)$ are also plotted against $N$, for all three types of covariance estimators $\hSig$.  We will observe from the graphs that portfolios with larger $c$ are exposed to have higher risks.

Finally, we fix the dimensionality $N=600$ and the number of simulation replications is now set to $500$. Values of two ratios are recorded, namely ratio of bounds
$$
\mathrm{RE}_1=\frac{\xi_T}{\hU(0.05)}=\frac{\|\bw\|_1^2\|\hSig-\Sig\|_{\max}}{2\sqrt{\hvar(\bw'\hSig\bw)}},
$$
and relative error
$$
\mathrm{RE}_2 = \frac{\hU(0.05)}{4\bw'\Sig\bw}=\frac{\sqrt{\hvar(\bw'\hSig\bw)}}{2\bw'\bSigma\bw}.
$$
This is computed for $c$ in a practical range of $[1, 2]$ and for several lengths of the time series. The means and standard deviations of the two ratios are summarized in tables below.

\subsection{Results}
In Figure~\ref{fig:riskdomain}, averages of annualized true risks $R(\bw)$ of 10000 portfolios $50$ sets of model parameters are plotted against dimensionality $N$. Multiple curves with different settings on $c$ are produced for comparison purpose. As shown in the figure, average of the actual risk ranges from less than 30\% to around 50\% per annum, as $c$ varies from 1 to 4 and $N$ gradually grows from 20 to 600.
\begin{figure}[ht]
\centering
\caption{Averages of annualized risks $R(\bw)$ with $\|\bw\|_1=$ 1, 2, 3 and 4, over $10000$ portfolios. }
\includegraphics[scale=0.7,angle=270]{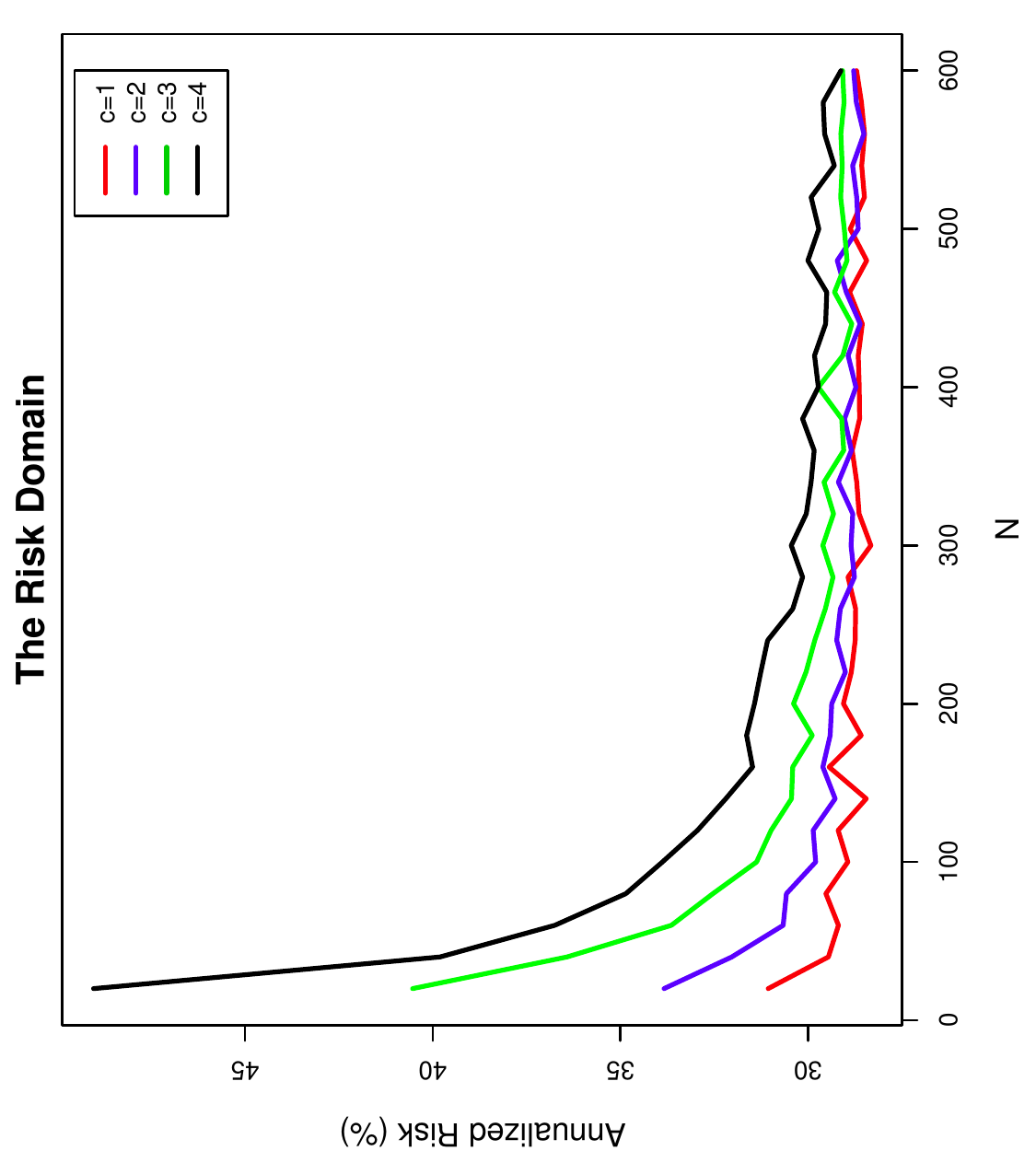}
\label{fig:riskdomain}
\end{figure}

The following two observations can be made from Figure~\ref{fig:riskdomain}:
\begin{itemize}
\item [(1)]
    The average risk is higher for a larger exposure parameter $c$. This is consistent with the fact that portfolios with greater gross exposure are more volatile, and hence incur higher risk.

\item [(2)]
    Given a gross exposure level $c$, as the portfolio size $N$ increases, the average risk decreases. The rate of decline is very fast until $N$ is around 150.  This is consistent with the theory that as $N$ increases, the portfolio becomes more diversified and the idiosyncratic risk is reduced through diversification.
\end{itemize}

\begin{figure}[ht]
 \centering
 \caption{Averages of $\Delta=|\bw'(\hSig-\bSigma)\bw|$ (blue curve), $\widehat{U}(0.05)=2\sqrt{\hvar(\bw'\hSig\bw)}$ (dashed curve) and $\xi_T = \|\bw\|_1^2\|\hSig-\bSigma\|_{\max}$ (red curve) for $c=1$ and 1.6 over 10,000 portifolios. }
\subfigure[c = 1]{
\includegraphics[scale=0.5,angle=270]{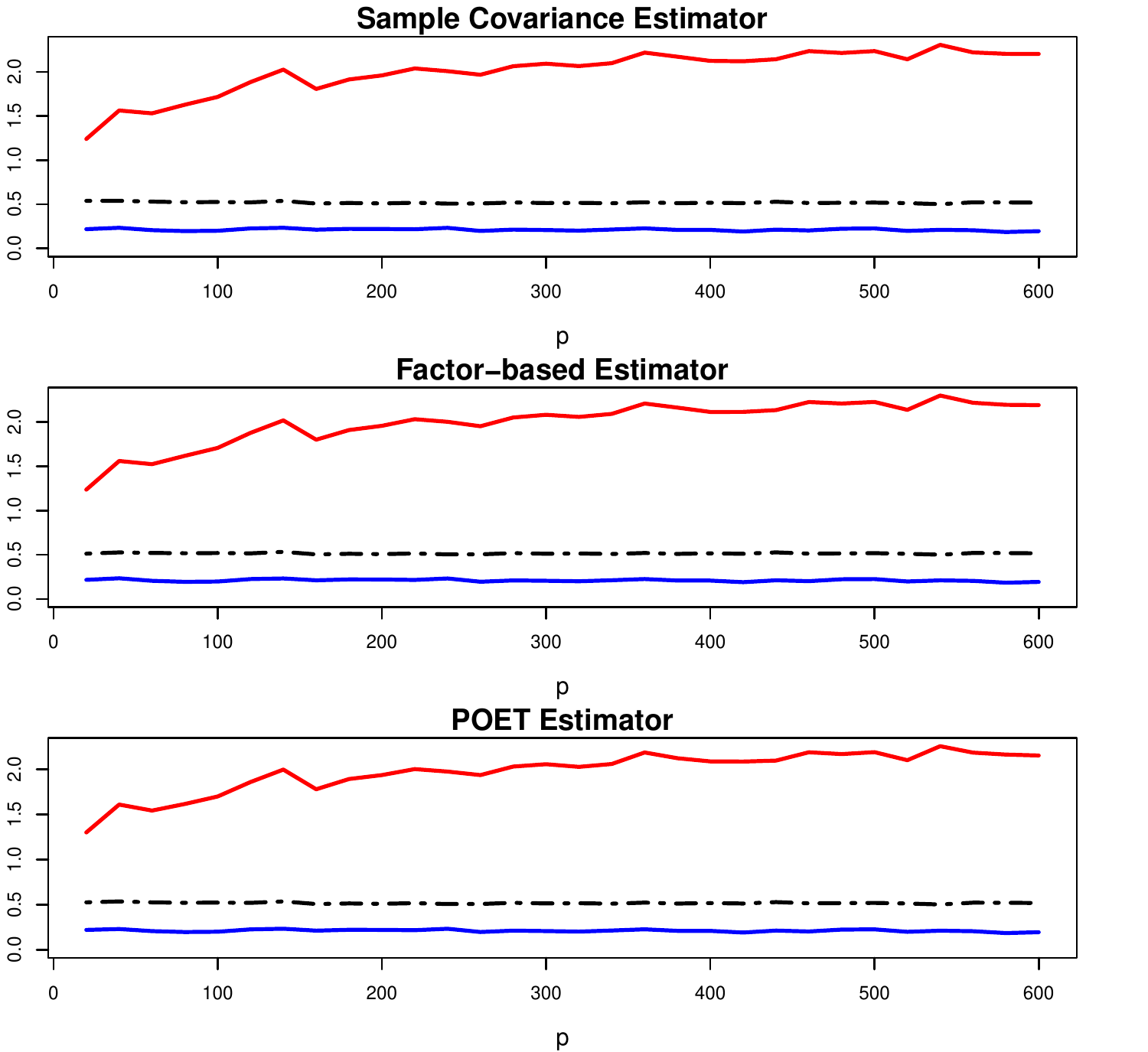}
\label{fig:c=1}
}
\subfigure[c = 1.6]{
\includegraphics[scale=0.5,angle=270]{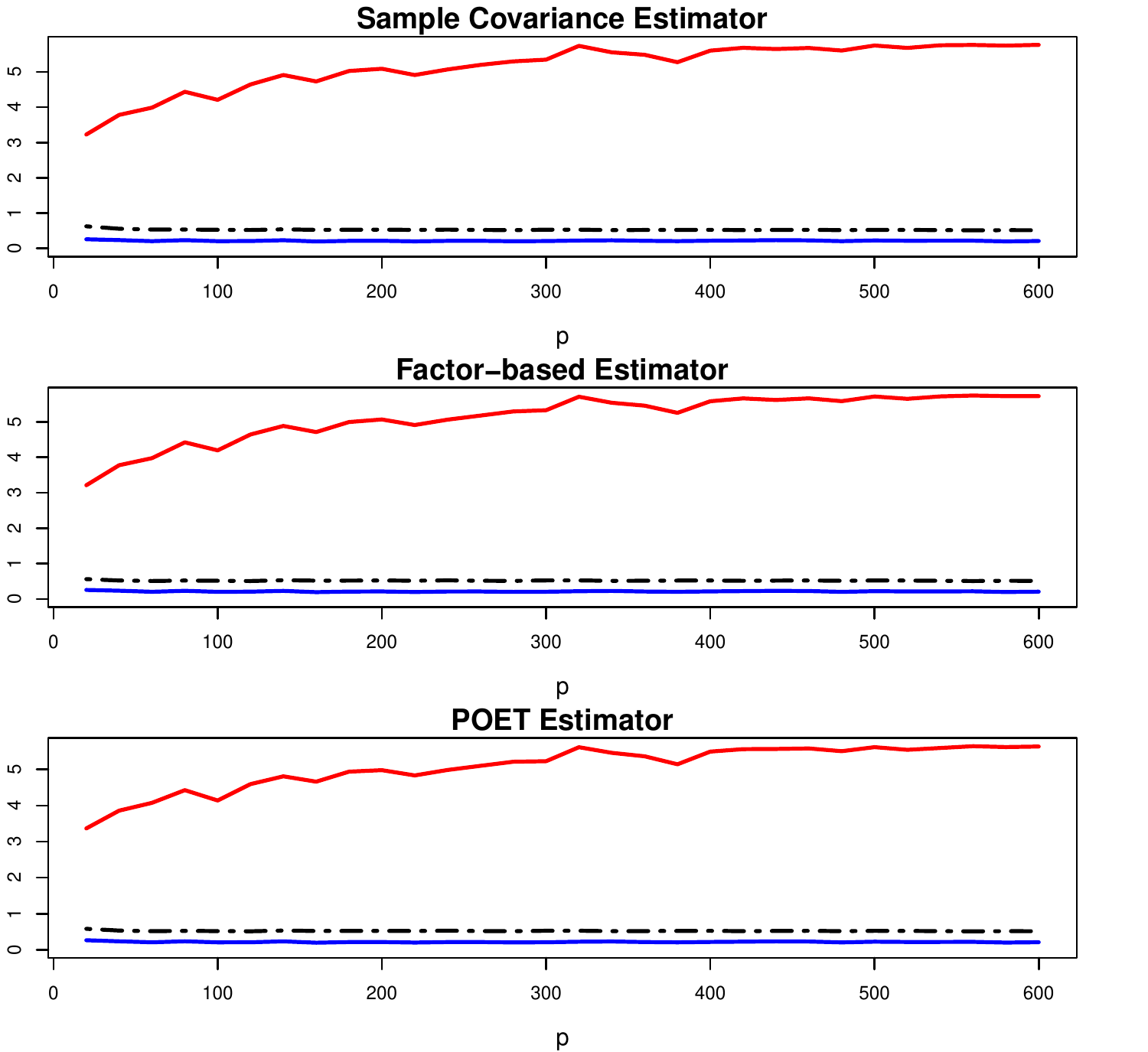}
\label{fig:c=1.6}
 }
\label{fig:2}
\end{figure}

\begin{figure}[ht]
\centering
\subfigure[c = 2]{
\includegraphics[scale=0.5,angle=270]{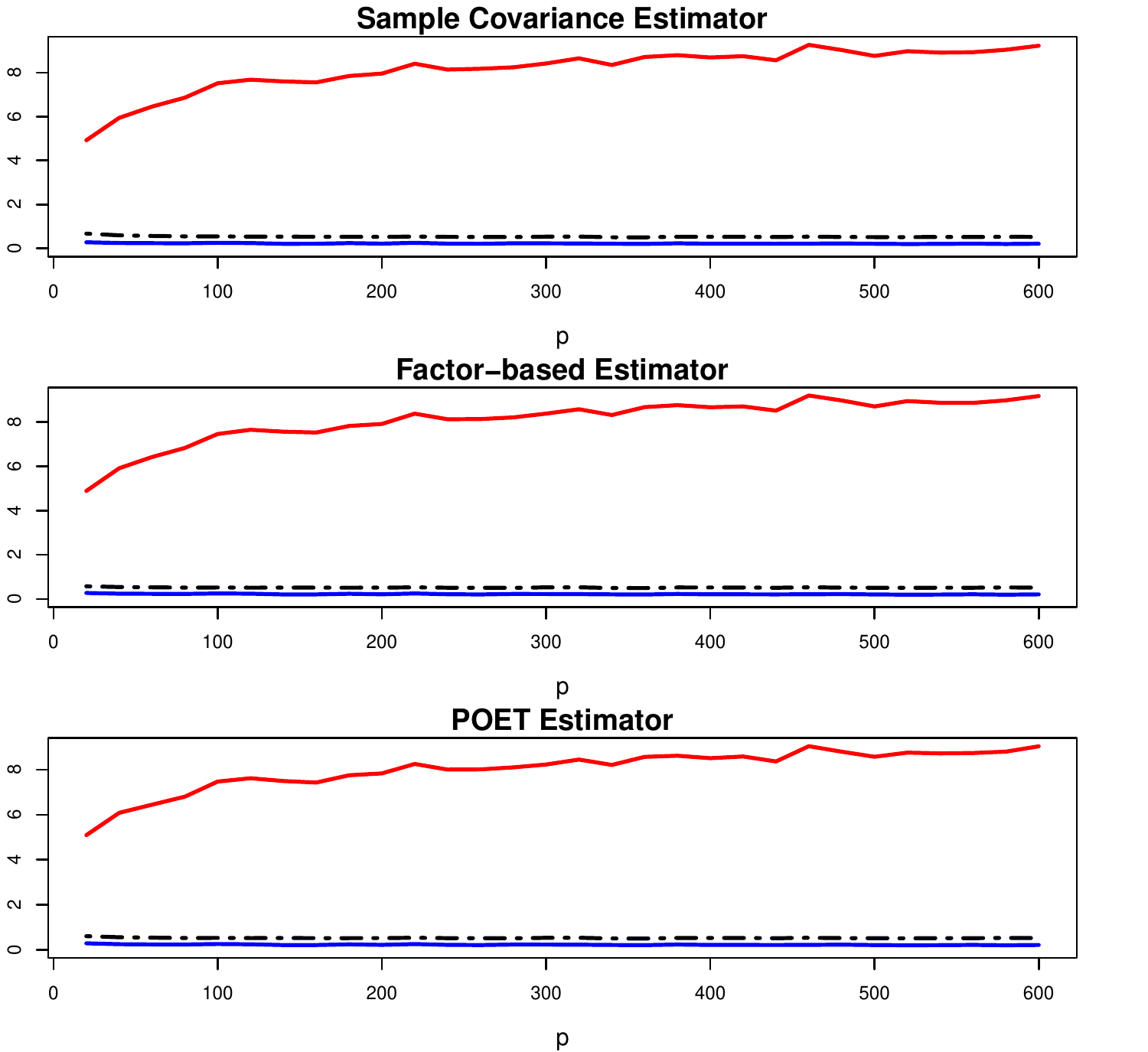}
\label{fig:c=2}
}
\subfigure[c = 3]{
\includegraphics[scale=0.5,angle=270]{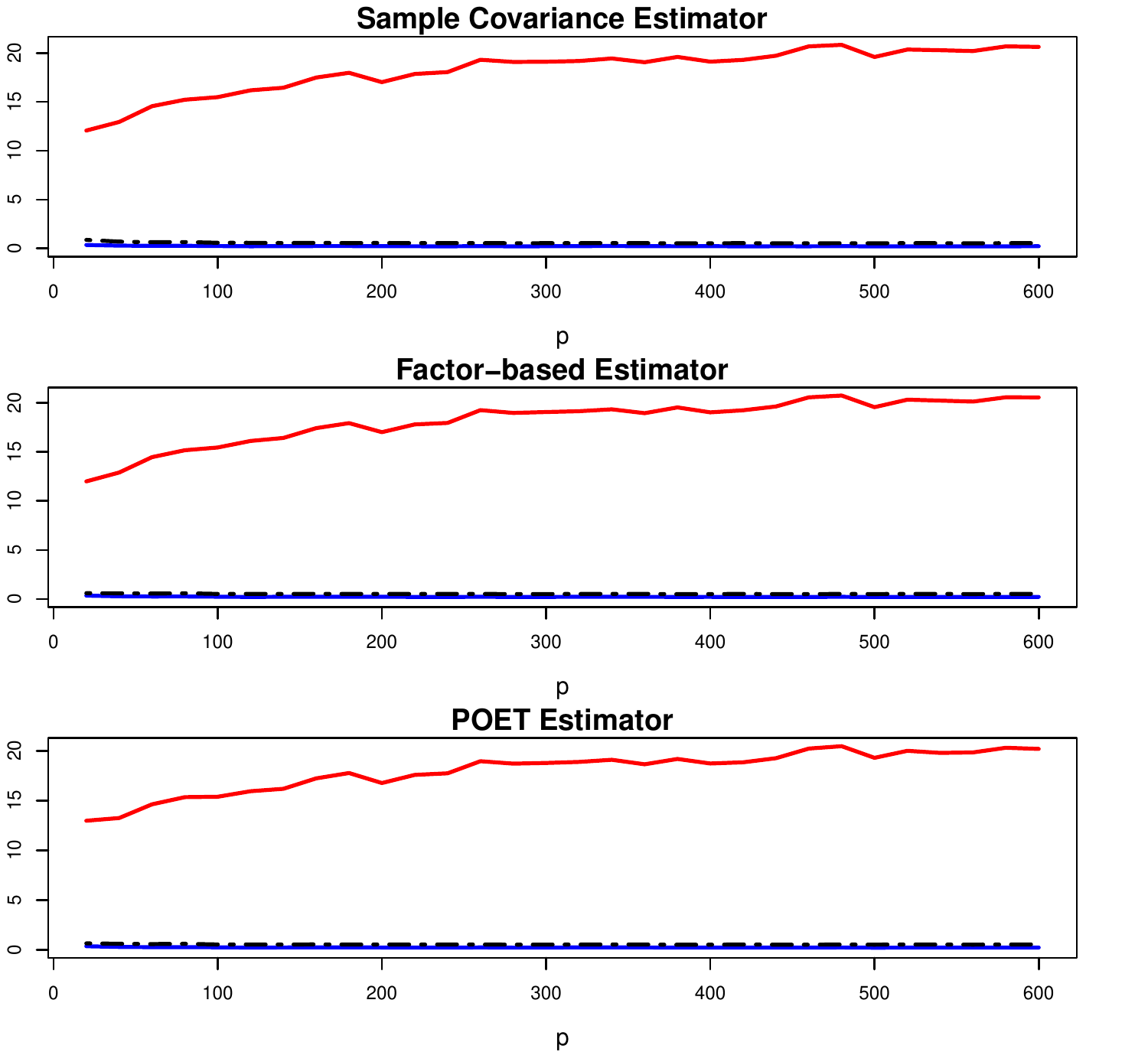}
\label{fig:c=3}
}

\caption{Same as in Figure~\ref{fig:2}, with $c$= 2 and 3.\label{fig:3}}
\end{figure}

In Figures~\ref{fig:2} and \ref{fig:3}, the average risk estimation errors are plotted along with their estimated error bounds for different exposure parameters $c$=1, 1.6, 2 and 3, using different estimators $\hSig=\bS, \hSig_f$ and $\hSig_P$. In particular, $c=1.6$ results in 130\% long positions and 30\% short positions (130/30 strategy). The 130/30 structure is popular in long-short funds.  In each of the small figure, the dashed curve corresponds to $\widehat{U}(0.05)$, the solid red curve corresponds to $\xi_T$ and the solid blue curve corresponds to $\Delta$. Based on these plots, we can observe the following features.

\begin{enumerate}
\item [(1)] The dashed curves lie entirely above the solid blue one, reflecting the validity of the 95\%-error bound of $\hU(0.05)$.

\item [(2)] The full confidence upper bound $\xi_T$ is indeed a very crude bound and is much larger than $\widehat{U}(0.05)$. The larger $c$ is, the worse the difference, which will be further detailed in Table~\ref{tab:3}.

\item [(3)]  H-CLUB  (dashed curve) slightly increases with larger $N$, but its degree of increases is much smaller than the crude bound $\xi_T$.
\end{enumerate}

Means and standard deviations (in parentheses) of $\mathrm{RE}_1=\xi_T/\widehat{U}(0.05) $ are summarized in Table~\ref{tab:3}, which quantifies the relation between the full confidence bound and the H-CLUB. Numerical results justify our observations in Figures~\ref{fig:2} and \ref{fig:3} in the sense that $\xi_T$ is in general many times greater than $\hU(\tau)$. Moreover, the ratio $\xi_T/\hU(\tau)$ increases dramatically as the gross exposure $\|\bw\|_1$ increases.

\begin{table}[h]
\centering
\caption{Averages and standard deviations (in parenthesis) of $\mathrm{RE}_1$ over $500$ iterations using three different estimators.}
\begin{tabular}{c|cccccc}
\hline
&$c=1$&$c=1.2$&$c=1.4$&$c=1.6$&$c=1.8$&$c=2$\\
\hline
&&&&&\\
$\mathrm{RE}_1$ &5.1280 &7.4632 &10.4257 &12.7665 &16.7107 &20.7675\\
$\bS$ &(2.1303)& (3.2754)&  (4.3942)& (5.5150)&  (7.1332) &(9.0050)\\
\hline
&&&&&\\
$\mathrm{RE}_1$ &5.1294 &7.4764& 10.4544 &12.7822 &16.8100 &20.9012\\
$\hSig_f$ &(2.1630)& (3.3155)&  (4.4402)&  (5.5731)&  (7.2942)&  (9.2622)\\
\hline
&&&&&\\
$\mathrm{RE}_1$ &5.0910 &7.3989& 10.3536 &12.6485 &16.6076& 20.6935\\
$\hSig_P$&(2.1672)& (3.3350)&  (4.5094)&  (5.5913)&  (7.3239)&  (9.3091)\\
\hline
\end{tabular}
\label{tab:3}
\end{table}

Averages and standard deviations of relative error $\mathrm{RE}_2 = \sqrt{\hvar(\bw'\hSig\bw)}/2\bw'\bSigma\bw$ with two choices of the length of time series $T$ are summarized in Table~\ref{tab:4}-\ref{tab:6}. $\mathrm{RE}_2$ measures the accuracy of the perceived risk $\widehat{R}(\bw)^{\frac{1}{2}}$ with respect to the actual risk $R(\bw)^{\frac{1}{2}}$, indeed $\mathrm{RE}_2\approx\text{ASD}(\widehat{R}(\bw)^{\frac{1}{2}})/R(\bw)^{\frac{1}{2}}$. From the tables, it is  not difficult to observe that standard deviations are small when compared to their corresponding means. The results also show that the relative error are negligible, at around $3\%\sim5\%$, ensuring the estimate of $R(\bw)$ a high level of accuracy. More interestingly, we realize that this ratio is approximately independent of the gross exposure $c$ but sensitive to the length of the time series. $\mathrm{RE}_2$ steadily decreases as $T$ grows. In addition, the asymptotic standard deviation of the factor-based estimators ($\hSig_f$ and $\hSig_P$) are slightly smaller than that of the sample covariance based risk estimator.
\begin{table}[h]
\centering
\caption{Averages and standard deviations of $\mathrm{RE}_2$ over $500$ iterations, with $T=200$.}
\begin{tabular}{c|ccccccc}
\hline
&$c=1$&$c=1.2$&$c=1.4$&$c=1.6$&$c=1.8$&$c=2$\\
\hline
&&&&&\\
$\mathrm{RE}_2$ &4.8381\%& 4.8076\%& 4.6563\%& 4.7499\%& 4.7267\%& 4.7723\%\\
$\bS$ &(1.0015\%)& (1.0147\%)& (0.9368\%)&(0.9989\%)& (0.9894\%)& (0.9648\%)\\
\hline
&&&&&\\
$\mathrm{RE}_2$ &4.8264\%& 4.8038\%& 4.6409\%& 4.7316\%& 4.7150\%& 4.7458\%\\
$\hSig_f$ &(0.9971\%)& (1.0110\%)& (0.9385\%)& (0.9933\%)& (0.9835\%)& (0.9646\%)\\
\hline
&&&&&\\
$\mathrm{RE}_2$ &4.8305\%& 4.8055\%& 4.6443\%& 4.7350\%& 4.7133\%& 4.7478\%\\
$\hSig_P$&(0.9993\%)& (1.0104\%)& (0.9368\%)& (0.9922\%)& (0.9859\%) &(0.9624\%)\\
\hline
\end{tabular}
\label{tab:4}
\end{table}


\begin{table}[ht]
\centering
\caption{Averages and standard deviations of $\mathrm{RE}_2$ over $500$ iterations, with $T=400$.}
\begin{tabular}{c|ccccccc}
\hline
&$c=1$&$c=1.2$&$c=1.4$&$c=1.6$&$c=1.8$&$c=2$\\
\hline
&&&&&\\
$\mathrm{RE}_2$ &3.4773\%& 3.4840\%& 3.4836\%& 3.4857\%& 3.5176\%& 3.4846\%\\
$\bS$ &(0.5081\%)& (0.4936\%)& (0.4783\%)& (0.5303\%)& (0.5117\%)& (0.5602\%)\\
\hline
&&&&&\\
$\mathrm{RE}_2$ &3.4693\%& 3.4759\%& 3.4699\%& 3.4708\%& 3.5009\%& 3.4588\%\\
$\hSig_f$ &(0.5081\%)& (0.4976\%)& (0.4775\%)& (0.5217\%)& (0.5127\%)& (0.5621\%)\\
\hline
&&&&&\\
$\mathrm{RE}_2$ &3.4744\%& 3.4773\%& 3.4737\%& 3.4737\%& 3.5029\%& 3.4619\%\\
$\hSig_P$&(0.5092\%)& (0.4964\%)& (0.4783\%)& (0.5229\%)& (0.5132\%)& (0.5629\%)\\
\hline
\end{tabular}
\label{tab:6}
\end{table}

Finally, we also observe from Tables \ref{tab:3}-\ref{tab:6} that the asymptotic variances (reflected by $\hU(\tau)$) of the estimators based on known and unknown factors are almost the same, and slightly smaller than that of the sample covariance estimator.

\section{Empirical Studies}

We assess the performance of H-CLUB in a portfolio allocation. We use the daily excess returns of 100 industrial portfolios formed on the size and book to market ratio from the website of Kenneth French. The study period is from July $1^{st}$ 2008 to June $29^{th}$ 2012 ($T=1000$). At the end of each month the covariance matrix is estimated by three estimators, the sample covariance, the factor-based estimator, and the POET estimator, using daily returns of the preceding 12 months ($T=252$). In particular, we employ the Fama-French three-factor model to construct the factor-based estimator.  Two types of strategies are tested, namely the equally weighted portfolio, and the minimum variance portfolio. The optimal portfolios are constructed under different exposure constraints ($c=1$ and $c=1.6$). The equally weighted portfolio is given by
$	
\widehat{\bw} = \left(1/N,\cdots,1/N\right).	
$	
The minimum variance  portfolio is given by
$$
\widehat{\bw}=\argmin_{\bw'\bone=1,\|\bw\|_1=c}\bw'\widehat{\bSigma}\bw.
$$
Portfolios are held for one month ($T=21$) and rebalanced at the beginning of the next month. Their actual risks in the holding month for $\widehat{\bw}$ defined above are
$$
R(\widehat{\bw}) = \left(\widehat{\bw}'\Sig \widehat{\bw}\right)^{1/2}, \quad \Sig=\frac{1}{21}\sum_{t=1}^{21}\by_t\by_t'.
$$
This is aggregated over the entirely testing period.

\begin{table}[ht]
\centering
\caption{True risk errors and estimated risk errors based on the 100 Fama-French Industrial Portfolios.\protect\\}
\begin{tabular}{l*{5}{c}}
\hline
\hline
&Average of& Average of &Average of&True &Estimated\\
Strategy&$\Delta(\times10^{-4})$ &$\widehat{U}(0.01)(\times10^{-4})$&True Risk&Risk Error& Risk Error\\
\hline
\multicolumn{6}{c}{{\it Sample Covariance Matrix Estimator}}\\
\hline
Equal weighted&2.356&2.757&20.81\%&11.18\%&11.37\%\\
Min variance ($c=1$)&1.006&1.232&14.38\%&7.00\%&7.45\%\\
Min variance ($c=1.6$)&0.497&0.622&11.58\%&4.69\%&5.18\%\\
\hline
\multicolumn{6}{c}{{\it Factor-Based Covariance Matrix Estimator}}\\
\hline
Equal weighted&2.352&2.693&20.81\%&11.16\%&11.22\%\\
Min variance ($c=1$)&0.999&1.234&14.45\%&6.95\%&7.48\%\\
Min variance ($c=1.6$)&0.475&0.607&11.79\%&4.52\%&5.07\%\\
\hline
\multicolumn{6}{c}{{\it POET Estimator}}\\
\hline
Equal weighted&2.353&2.757&20.81\%&11.17\%&11.38\%\\
Min variance ($c=1$)&1.005&1.171&14.38\%&6.99\%&7.07\%\\
Min variance ($c=1.6$)&0.490&0.572&11.59\%&4.61\%&4.64\%\\
\hline
\hline

\end{tabular}
\small\it
Here $\Delta = |\bw'(\bSigma-\widehat{\bSigma}){\bw}|$ and $\widehat{U}(0.01)=2.58(\widehat{\var}(\bw'\widehat{\bSigma}\bw))^{1/2}$.  True risk  is $R(\widehat{\bw})$. The True Risk Error and Estimated Risk Error are  $|(\bw'\bSigma\bw)^{1/2}-(\bw'\widehat{\bSigma}\bw)^{1/2}|$ and $\widehat{U}(0.01)/\sqrt{4\bw'\widehat{\bSigma}\bw}$  respectively.

\label{tab:7}
\end{table}

For each covariance matrix estimator and strategy, we study five quantities, whose respective averages over the whole study period are summarized in Table~\ref{tab:7}. In particular,  the estimated risk error $\widehat{U}(0.01)/\sqrt{4\bw'\widehat{\bSigma}\bw}$ is the H-CLUB for the true risk error.  See, for example, Corollalry~\ref{co31}. 
Here the risks are annualized. By comparing the first two columns in Table~\ref{tab:7}, we observe that $\widehat{U}(0.01)$ is uniformly greater than $\Delta$, regardless of the strategies and the covariance matrix estimators. This is in line with the expectation that $\widehat{U}(0.01)$ is a 99\% upper bound of the estimation error of portfolio variances. 
Moreover, as shown in the two rightmost columns, results are satisfactory in the sense that the estimated risk errors are close ($<1\%$ per annum) to the true risk error.

\section{Conclusions}
In this paper we address the estimation and assessment for the risk of a large portfolio. The risk is estimated by a substitution of a good estimator of the volatility matrix.   We propose factor-based risk estimators, based on  the approximate factor model with known factors and unknown factors.  For the first time in the literature,  we derive the limiting distribution of the estimated risks under high dimensionality.

Given that the existing upper bound for the risk estimation error is too crude and not applicable in practice, we  introduce a new method, H-CLUB,  to assess the accuracy of the risk estimation based on the confidence intervals.  Our numerical results demonstrate that the proposed upper bounds significantly outperform the traditional crude bounds, and provide insightful assessment of the estimation of the true portfolio risks.

It is demonstrated in the empirical study that the financial excess returns may not be globally stationary.  Our method also allows for  locally stationary  time series and can also allow slow-time varying covariance matrices  through localization in time (time-domain smoothing).  

\appendix

\section{Proofs for the Sample Covariance}

Define $Z_{T,t}=\bw'\bR_t\bR_t'\bw-E\bw'\bR_t\bR_t'\bw$, where $Z$ depends on $T$ through $\dim(\bR_t)=N=N_T$ and allocation vector $\bw$. Then $\gamma_T(h)=EZ_{T, t}Z_{T, t+h}$.   In particular, $\gamma_T(0)=\var(Z_{T,t}).$

\subsection{Proof of  Lemma \ref{l31}}

\begin{lem}\label{la.2}
(i) $|(\bw'\bS\bw)^2-(\bw'\Sig\bw)^2|=O_p(T^{-1/2}  \sigma_T).$\\
(ii) $\max_{h\leq L}|T^{-1}\sum_{t=1}^{T-h}(\bw'\bR_t)^2(\bw'\bR_{t+h})^2-E(\bw'\bR_t)^2(\bw'\bR_{t+h})^2|=O_p(\sqrt{L/T}).$\\
(iii) $\max_{h\leq L}|\bw'\bS\bw-T^{-1}\sum_{t=1}^{T-h}(\bw'\bR_t)^2|=O_p(L^2\bw'\Sig\bw/T)$.\\
(iv) $\max_{h\leq L}|\bw'\bS\bw-T^{-1}\sum_{t=1}^{T-h}(\bw'\bR_{t+h})^2|=O_p(L^2\bw'\Sig\bw/T)$.
\end{lem}
\begin{proof}
Note that for any $N\times N$ matrix $\bA=(a_{ij})$, $|\bw'\bA\bw|\leq \|\bA\|_{\max} \|\bw\|_1^2$. Thus $|(\bw'\bS\bw)^2-(\bw'\Sig\bw)^2|=O_p(|\bw'(\bS-\Sig)\bw|)=O_p(|T^{-1}\sum_{t=1}^TZ_{T,t}|)$. The Chebyshev inequality implies $|T^{-1}\sum_{t=1}^TZ_{T,t}|= O_p(T^{-1/2}\sqrt{  \sigma^2_T   })$.\\
(ii) Let $X_{t,h}=(\bw'\bR_t)^2(\bw'\bR_{t+h})^2$.  By the Chebyshev inequality,  for any $s>0,$
$$
P(\max_{h\leq L}|\frac{1}{T}\sum_{t=1}^TX_{t,h}-EX_{t,h}|>s)\leq L\max_{h\leq L}P(|\frac{1}{T}\sum_{t=1}^TX_{t,h}-EX_{t,h}|>s)\leq\frac{L\max_{h\leq L}\var(\sum_{t=1}^TX_{t,h})}{T^2s^2}.
$$
Note that
$\max_{h\leq L}\var(\sum_{t=1}^TX_{t,h})=O(T)$ since $\max_{h\leq L}\var(X_{t, h})=O(1)$ and \\
$\max_{h\leq L}\sum_{t=1}^T\cov(X_{1,h}, X_{t+1,h})=O(1)$. Therefore, for arbitrarily small $\epsilon>0$, by choosing  $s>\sqrt{LM/(\epsilon T)}$, $P(\max_{h\leq L}|\frac{1}{T}\sum_{t=1}^TX_{t,h}-EX_{t,h}|>s)<\epsilon$, which implies $\max_{h\leq L}|\frac{1}{T}\sum_{t=1}^TX_{t,h}-EX_{t,h}|=O_p(\sqrt{L/T})$.  The conclusion then follows from the adjustment of the $L$ terms in the summation.

(iii) The left hand side is   $\max_{h\leq L}T^{-1}\sum_{t=T-h+1}^T(\bw'\bR_t)^2=\max_{T-L+1\leq t\leq T}(\bw'\bR_t)^2L/T$. For any $s>0$,  $P(\max_{T-L+1\leq t\leq T}(\bw'\bR_t)^2>s)\leq LP((\bw'\bR_t)^2>s)\leq L\bw'\Sig\bw/s$, which then implies $\max_{T-L+1\leq t\leq T}(\bw'\bR_t)^2=O_p(L\bw'\Sig\bw)$. The desired result then follows.

(iv) A similar argument as above shows $\max_{T+1\leq t\leq T+L}(\bw'\bR_t)^2=O_p(L\bw'\Sig\bw)$. Hence\\
$\max_{h\leq L}T^{-1}\sum_{t=T-h+1}^T(\bw'\bR_{t+h})^2\leq \max_{T+1\leq t\leq T+L}(\bw'\bR_t)^2L/T=O_p(L^2\bw'\Sig\bw/T).$ This implies that the desired quantity is bounded by $a+O_p(L^2\bw'\Sig\bw/T)$ where
$$
a=\max_{h\leq L}|\frac{1}{T}\sum_{t=1}^{T}[(\bw\bR_t)^2-(\bw'\bR_{t+h})^2]|\leq |\frac{1}{T}\sum_{ t=1}^L(\bw'\bR_t)^2|+ |\frac{1}{T}\sum_{ t=1}^L(\bw'\bR_{T+t})^2|.
$$ Note that $|\frac{1}{T}\sum_{ t=1}^L(\bw'\bR_t)^2|\leq\max_{1\leq t\leq L}(\bw'\bR_t)^2L/T=O_p(L^2\bw'\Sig\bw/T)$. Similarly we have $|\frac{1}{T}\sum_{ t=1}^L(\bw'\bR_{T+t})^2|=O_p(L^2\bw'\Sig\bw/T).$
\end{proof}

\begin{lem}\label{la.3}
$\max_{h\leq L}|\hgamma(h)-\gamma_T(h)|=O_p(\sqrt{L/T})$.
\end{lem}
\begin{proof}
The triangular inequality implies $\max_{h\leq L}|\hgamma(h)-\gamma_T(h)|\leq \sum_{i=1}^4a_i$, where
\begin{eqnarray*}
a_1=\max_{h\leq L}|T^{-1}\sum_{t=1}^{T-h}(\bw'\bR_t)^2(\bw'\bR_{t+h})^2-E(\bw'\bR_t)^2(\bw'\bR_{t+h})^2|, a_2= |(\bw'\bS\bw)^2-(\bw'\Sig\bw)^2|\cr
a_3=\bw'\bS\bw\max_{h\leq L}|\bw'\bS\bw-T^{-1}\sum_{t=1}^{T-h}(\bw'\bR_t)^2|, a_4=\bw'\bS\bw\max_{h\leq L}|\bw'\bS\bw-T^{-1}\sum_{t=1}^{T-h}(\bw'\bR_{t+h})^2|.
\end{eqnarray*}
We have, $\bw'\bS\bw\leq |\bw'(\bS-\Sig)\bw|+\bw'\Sig\bw=O_p(\bw'\Sig\bw +T^{-1/2} \sigma^2_T )$. It then follows from Lemma \ref{la.2} and  $  \sigma^2_T   =O(1)$, $L^{3}=O(T)$,  $\bw'\Sig\bw=O(1)$ that $a_i=O_p(\sqrt{L/T})$ for $i=1...4$, which implies $\max_{h\leq L}|\hgamma(h)-\gamma_T(h)|=O_p(\sqrt{L/T})$.
\end{proof}

\textbf{Proof of  Lemma \ref{l31}}

By the triangular inequality, $|\hsig^2-  \sigma^2_T   |\leq\sum_{i=1}^3b_i$, where
$$
b_1=|\hgamma(0)-\gamma_T(0)|, \quad b_2=2\sum_{h=1}^L|\hgamma(h)-\gamma_T(h)|, \quad b_3=2\sum_{h>L}\gamma_T(h)
$$
Here $b_2\leq 2L\max_{h\leq L}|\hgamma(h)-\gamma_T(h)|=O_p(L\sqrt{L/T}).$
The convergence rate then follows from Lemma \ref{la.3}.  The second part $\hU_S(\tau)=o(\sqrt{\log N/T})$ is due to $\hsig^2=O_p(\sigma_T^2)$ as $|  \sigma^2_T   -\hsig^2|=o_p(  \sigma^2_T   )$ and $\sigma_T^2=O(1)=o(\log N)$, as $N\rightarrow\infty.$

\subsection{Proof of Theorem \ref{t31}}

\begin{lem} (i) $EZ_{T,1}^2=O(1)$ and $\max_{l\leq T}|\gamma_T(l)|=O(1)$.\\
(ii) For any $K\in[m, T]$,   $\var(\sum_{t=1}^KZ_{T,t})= K\gamma_T(0)+2K\sum_{h=1}^K(1-h/K)\gamma_T(h)=O(K)$.

\end{lem}

\begin{proof}
(i) It suffices to show $E(\bw'\bR_t)^4=O(1)$. In fact by $ \max_{i\leq N}ER_{it}^4=O(1)$, \\
 $E(\bw'\bR_t)^4=\sum_{ijkl=1}^Nw_iw_jw_kw_lER_{it}R_{jt}R_{kt}R_{lt}\leq \max_{i\leq N}ER_{it}^4\|\bw\|_1^4=O(1)$.   The second part follows immediately.\\
(ii) It is well known that for a stationary process with zero mean, $\var(K^{-1}\sum_{t=1}^KZ_{T,t})=K^{-1}\gamma_T(0)+2K^{-1}\sum_{h=1}^K(1-h/K)\gamma_T(h)$, which implies the result.

\end{proof}

\begin{lem} Under the assumptions of Theorem \ref{t31},
 \begin{equation}\label{eqa1}
\left[\var\left(\sum_{t=1}^T(\bw'\bR_t)^2\right)\right]^{-1/2}   T\bw'(\bS-\Sig)\bw\rightarrow^d\mathcal{N}(0,1).
\end{equation}
\end{lem}

\begin{proof}
The proof is based on Theorem 2.1 of Peligrad (1996).  We have $\sqrt{T}\bw'(\bS-\Sig)\bw=T^{-1/2}\sum_{t=1}^TZ_{T,t}$. Define $B_{T, K}^2=\var(\sum_{t=1}^KZ_{T,t})$ and $B_{T}^2=\var(\sum_{t=1}^TZ_{T,t})=O(T)$. Also let $  \sigma^2_T   =\gamma_T(0)+2\sum_{h=1}^{\infty}\gamma_T(h)$.   By Davydov's inequality (Proposition 2.5 of Fan and Yao, 2003 with $p=1/2$ and $q=1/4$),  there are  constants $M, M_1, M_2>0$ such that for any integer $h$,
 $$
| \gamma_T(|h|)|\leq 8\alpha_R(|h|)^{1/4}(E(\bw'\bR_t)^2)^{1/2}(E(\bw'\bR_t)^4)^{1/4}=M_2\exp(-M|h|^{r_3}/4) $$
where the last equality follows from the $\alpha$-mixing condition and that $E(\bw'\bR_t)^4=O(1)$. By the assumption that $\alpha_R(T)=o(\gamma_T(0))$, the correlation $|$Corr$(Z_{T,t}, Z_{T,t+T})|\leq |\gamma_T(T)|/\gamma_T(0)=o(1)$. Moreover, the Lindeberg condition holds given $\max_{i\leq N}ER_{it}^8<\infty$. Hence  the conditions of Theorem 2.1 of Peligrad (1996) are  satisfied, which implies $B_T^{-1}\sum_{t=1}^TZ_{T,t}\rightarrow^d\mathcal{N}(0,1)$, equivalent to (\ref{eqa1}).

\textbf{Proofs of Theorem \ref{t31} and Corollary \ref{co31}}

Now let $\xi(T)=-2\sum_{h=1}^{T}h\gamma_T(h)/T$.  By the assumption that  $\xi(T)=o(  \sigma^2_T   )$,
 we have $T^{-1/2}(  \sigma^2_T   )^{-1/2}\sum_{t=1}^TZ_{T,t}\rightarrow^d \mathcal{N}(0,1)$. This also implies
 \begin{equation} \label{ea1}
\sqrt{\frac{T}{  \sigma^2_T   }}\bw'(\bS-\Sig)\bw\rightarrow^d \mathcal{N}(0, 1).
\end{equation}

Due to the assumptions that $L^{3/2}T^{-1/2}=o(  \sigma^2_T   )$ and $\sum_{h>L}\gamma_T(h)=o(  \sigma^2_T   )$, and Lemma \ref{l31},   we have  $|  \sigma^2_T   -\hsig^2|=o_p(  \sigma^2_T   )$.  Since $\bw'(\bS-\Sig)\bw=O_p(T^{-1/2}\sqrt{  \sigma^2_T   }),$
$$
\sqrt{T}|\bw'(\bS-\Sig)\bw|\left|\frac{1}{\sqrt{  \sigma^2_T   }}-\frac{1}{\sqrt{\hsig^2}}\right|=o_p(1).
$$
It then follows from  (\ref{ea1}) that $
\sqrt{T/\hsig^2}\bw'(\bS-\Sig)\bw\rightarrow^d \mathcal{N}(0, 1),
$
which gives the H-CLUB.  Corollary \ref{co31} follows straightforward from applying the delta method.


\end{proof}

\section{Proofs for the Factor-based Estimation}

\subsection{Proof of Lemma \ref{l32}}

 \begin{lem}\label{la.9}
$\max_{h\leq L}|\hgamma_f(h)-\gamma_f(h)|=O_p(\sqrt{(L+\log N)/T}).$
\end{lem}
\begin{proof}

The triangular inequality implies $\max_{h\leq L}|\hgamma_f(h)-\gamma_f(h)|\leq\sum_{i=1}^4a_i$, where
\begin{eqnarray*}
a_1&=&\max_{h\leq L}|T^{-1}\sum_{t=1}^{T-h}(\bw'\hB\bff_{t+h})^2(\bw'\hB\bff_t)^2-E(\bw'\bB\bff_t)^2(\bw'\bB\bff_{t+h})^2|, \cr
a_2&= &|(\bw'\hB\hcov(\bff_t)\hB'\bw)^2-(\bw'\bB\cov(\bff_t)\bB'\bw)^2|\cr
a_3&=&\bw'\hB\hcov(\bff_t)\hB'\bw\max_{h\leq L}|\bw'\hB\hcov(\bff_t)\hB'\bw-T^{-1}\sum_{t=1}^{T-h}(\bw'\hB\bff_t)^2|, \cr
a_4&=&\bw'\hB\hcov(\bff_t)\hB'\bw\max_{h\leq L}|\bw'\hB\hcov(\bff_t)\hB'\bw-T^{-1}\sum_{t=1}^{T-h}(\bw'\hB\bff_{t+h})^2|.
\end{eqnarray*}
$a_1$ is bounded by $a_{11}+a_{12}$, where \\
 $a_{11}=
 \max_{h\leq L}|T^{-1}\sum_{t=1}^{T-h}(\bw'\bB\bff_{t+h})^2(\bw'\bB\bff_t)^2-E(\bw'\bB\bff_t)^2(\bw'\bB\bff_{t+h})^2|,
 $
and \\$a_{12}=\max_{h\leq L}|T^{-1}\sum_{t=1}^{T-h}(\bw'\hB\bff_{t+h})^2(\bw'\hB\bff_t)^2-(\bw'\bB\bff_{t+h})^2(\bw'\bB\bff_t)^2|$.

Given the assumption that $\max_{h\leq L}\sum_{t=1}^T\cov[(\bw'\bB\bff_1)^2(\bw'\bB\bff_{1+h})^2, (\bw'\bB\bff_{1+t})^2(\bw'\bB\bff_{1+t+h})^2]=O(1)$, the same argument of the proof of Lemma \ref{la.2}(ii) implies $a_{11}=O_p(\sqrt{L/T}).$  On the other hand, by (B.14) of Fan et al. (2011),  $\|\hB-\bB\|_{\max}=O_p(\sqrt{\log N/T})$, which implies $\|\bw'(\hB-\bB)\|=O_p(\sqrt{\log N/T})$. It is then easy to show that $a_{12}=O_p(\sqrt{\log N/T})$. It follows that $a_1=O_p(\sqrt{(L+\log N)/T}).$
By the triangular inequality, $a_2=O_p(\sqrt{\log N/T})$. Finally, by the same argument of the proof of Lemma \ref{la.2}, we have $a_3=O_p(L^2/T)=a_4$.
\end{proof}

\textbf{Proof of Lemma \ref{l32}}

We have $|\hsig_f^2-  \sigma^2_f   |\leq\sum_{i=1}^3b_i$, where
$
b_1=|\hgamma_f(0)-\gamma_f(0)|, $$ b_3=2\sum_{h>L}\gamma_f(h)
$, and $ b_2=2\sum_{h=1}^L|\hgamma_f(h)-\gamma_f(h)|.$
Lemma \ref{la.9} implies $b_2\leq 2L\max_{h\leq L}|\hgamma_f(h)-\gamma_f(h)|=O_p(L\sqrt{(L+\log N)/T})$, which gives the convergence rate.  The second statement  is due to $\hsig_f^2=o_p( \log N)$.

\subsection{Proof of Theorem \ref{t32}}\label{sb2}

 Write $\bR=(\bR_1,...,\bR_T)$ be $N\times T$; $\bF=(\bff_1,...,\bff_T)$ be $r\times T$, and $\hcov(\bff_t)=\bF\bF'/T.$ We have $\hB=\bR\bF'(\bF\bF')^{-1}$.
 Define $\bC_T=\hB-\bB$ and $\bD_T=\hcov(\bff_t)-\cov(\bff_t).$ The we have the following decomposition:
$\bw'(\hSig_f-\Sig)\bw=\sum_{i=1}^4d_i,
$
where
\begin{eqnarray*}
d_1=\bw'\bB\bD_T\bB'\bw; \quad d_2=2\bw'\bC_T\hcov(\bff_t)\bB'\bw;
\cr
d_3=\bw'\bC_T\hcov(\bff_t)\bC_T'\bw,\quad d_4=\bw'(\hSig_u-\Sig_u)\bw.
\end{eqnarray*}
 We now study each of the above four terms separately. Let $\bE=(\bu_1,...,\bu_T)$ be $N\times T$. Then  $\bC_T=\bE\bF'(\bF\bF')^{-1}$.

\begin{lem}  \label{la.4} (i) $\|\bF\bE'\bw\|=O_p(T^{1/2}(\bw'\Sig_u\bw)^{1/4}(E|\bw'\bu_t|^4)^{1/8})$.\\
(ii) $|d_2|=O_p(T^{-1/2}(\bw'\Sig_u\bw)^{1/4}(E|\bw'\bu_t|^4)^{1/8})  )$.
\end{lem}
\begin{proof}  We have,
\begin{eqnarray*}
E\|\bF\bE'\bw\|^2&=&E[\tr(\bw'\bE\bF'\bF\bE'\bw)]=\tr[E(\bF\bE'\bw\bw'\bE\bF')]=\tr[E(\bF E(\bE'\bw\bw'\bE|\bF)\bF')]\cr
&=&\tr[E(\bF E(\bE'\bw\bw'\bE)\bF')].
\end{eqnarray*}
For the inner expectation, $
E(\bE'\bw\bw'\bE)=(E[\bu_t'\bw\bw'\bu_s])_{t\leq t, s\leq T}=( \cov(\bw'\bu_t, \bw'\bu_s)  )_{t\leq t, s\leq T}.
$
By Davydov's inequality, (see, e.g., Proposition 2.5 of Fan and Yao, 2003 with $p=1/2$ and $q=1/4$),  $| \cov(\bw'\bu_t, \bw'\bu_s)|\leq 8 \alpha_f(|t-s|)^{1/4}(\bw'\Sig_u\bw)^{1/2}(E|\bw'\bu_t|^4)^{1/4}$, where $\alpha(\cdot)$ denotes the $\alpha$-mixing coefficient. By $\sum_{t=1}^{\infty}\alpha_f(t)^{1/4}<\infty$, we have
\begin{eqnarray*}
E\|\bF\bE'\bw\|^2&=&\sum_{k=1}^r\sum_{t=1}^T\sum_{s=1}^T\cov(\bw'\bu_t, \bw'\bu_s) E(f_{kt}f_{ks})\cr
&=&O(1)(\bw'\Sig_u\bw)^{1/2}(E|\bw'\bu_t|^4)^{1/4}
\sum_{t=1}^T\sum_{s=1}^T\alpha_f(|t-s|)^{1/4}=O(T(\bw'\Sig_u\bw)^{1/2}(E|\bw'\bu_t|^4)^{1/4}),
\end{eqnarray*}
which then implies (i). For part (ii), we have$$
|d_2|=2|\bw'\bB\hcov(\bff_t)(\bF\bF')^{-1}\bF\bE'\bw|=\frac{2}{T}|\bw'\bB\bF\bE'\bw|\leq \frac{2}{T}\|\bw'\bB\|\|\bF\bE'\bw\|.
$$
Now write $\bB=(\lambda_{ij})_{i\leq N, j\leq r}$, then $\|\bw'\bB\|^2=\sum_{j=1}^r(\sum_{i=1}^Nw_i\lambda_{ij})^2\leq {\max_{i,j}}|\lambda_{ij}|r\|\bw\|_1^2=O(1)$.

\end{proof}

\begin{lem} \label{lb.2}For the factor-based thresholded error covariance matrix,
$$
\|\hSig_u-\Sig_u\|=O_p\left(s_N\left(\frac{\log N}{T}\right)^{1/2-q/2}\right)
$$
\end{lem}
\begin{proof}
By Lemma 3.1 in Fan et al. (2011), we have, $\max_{i\leq N}T^{-1}\sum_{t=1}^T(\hat{u}_{it}-u_{it})^2=O_p({\log N/T})$.  The result then follows from Theorem A.1 in Fan et al. (2013).
\end{proof}

\begin{lem}
\label{la.6}
(i) $|d_3|=O_p(T^{-1}(\bw'\Sig_u\bw)^{1/2}(E|\bw'\bu_t|^4)^{1/4})$.\\
(ii) $|d_4|=O_p(s_N(\log N/T)^{1/2-q/2}\bw'\Sig_u\bw)$
\end{lem}
\begin{proof} (i) Because $\|(\bF\bF')^{-1}\|=O_p(T^{-1})$,\\
$
|d_3|=T^{-1}\bw'\bE\bF'(\bF\bF')^{-1}\bF\bE'\bw=O_p(T^{-2}\| \bF\bE'\bw \|^2).
$
It then follows from Lemma \ref{la.4}.\\
(ii) it follows from $|d_4|\leq \|\hSig_u-\Sig_u\|\|\bw\|^2\leq\lambda_{\min}^{-1}(\Sig_u)\|\hSig_u-\Sig_u\|\bw'\Sig_u\bw$ and Lemma \ref{lb.2}.
\end{proof}

 \begin{lem}\label{la.7}
$ \sum_{h=1}^{\infty}|\gamma_f(h)|<\infty$
 \end{lem}

 \begin{proof}
 By Davydov's inequality (Proposition 2.5 of Fan and Yao, 2003 with $p=1/2$ and $q=1/4$),  there are  constants $M_1, M_2>0$ such that for any integer $h$,
 $$
| \gamma_f(|h|)|\leq 8 \alpha_f(|h|)^{1/4}(E(\bw'\bB\bff_t)^2)^{1/2}(E(\bw'\bB\bff_t)^4)^{1/4}=M_2\exp(-M|h|^{r_3}/4) $$
where the last equality follows from the $\alpha$-mixing condition as well as the fact that $E(\bw'\bB\bff_t)^4=O(1)$ due to $\|\bw\|_1=O(1).$ The result the follows from $\sum_{h=1}^{\infty}\exp(-Ch^{r_3})<\infty $ for any $C, r_3>0.$
 \end{proof}

 \begin{lem}\label{la.8}
 $\sqrt{T/\sigma_f^2}d_1\rightarrow^d \mathcal{N}(0,1)$.
 \end{lem}

\begin{proof}
Let $Z_{T,t}=\bw'\bB (\bff_t\bff_t'-E\bff_t\bff_t')\bB'\bw$, which depends on $T$ through $\dim(\bw)=N_T$. Hence $d_1=T^{-1}\sum_{t=1}^TZ_{T,t}$.  Note that $\|\bw'\bB\|^2\leq r\|\bB\|_{\max}^2\|\bw\|_1^2=O(1)$. Hence $EZ_{T,1}^2=O(1)$. Similar to the proof of Theorem 3.1,  we define $B_{T, K}^2=\var(\sum_{t=1}^KZ_{T,t})$ and $B_{T}^2=\var(\sum_{t=1}^TZ_{T,t})=O(T)$.  By the assumption that $\alpha_f(T)=o(\gamma_f(0))$, the correlation $|$Corr$(Z_{T,t}, Z_{T,t+T})|\leq |\gamma_f(T)|/\gamma_f(0)=o(1)$. Moreover, the Lindeberg condition holds given the exponential tail of $\bff_t$. Hence  the conditions of Theorem 2.1 of Peligrad (1996) are  satisfied, which implies
\begin{equation}\label{ea2}
B_T^{-1}\sum_{t=1}^TZ_{T,t}\rightarrow^d\mathcal{N}(0,1).
\end{equation}


 For $\gamma_f(h)=\cov((\bw'\bB\bff_t)^2, (\bw'\bB\bff_{t+h})^2)$, we have $\gamma_f(h)=\cov(Z_{T,t}, Z_{T, t+h}).$
Now $B_T^2=T\gamma_f(0)+2T\sum_{h=1}^T\gamma_f(h)-2T\sum_{h=1}^Th\gamma_f(h)/T.$ Because $\sum_{h=1}^Th\gamma_f(h)/T=o(\gamma_f(0)+2\sum_{h=1}^{\infty}\gamma_f(h))$, we have $T^{-1/2}(\sigma_f^2)^{-1/2}\sum_{t=1}^TZ_{T,t}\rightarrow^d \mathcal{N}(0,1)$, where $\sigma_f^2=\gamma_f(0)+2\sum_{h=1}^{\infty}\gamma_f(h)$. The result then follows.
\end{proof}

\begin{lem}\label{lb.7}
 \begin{equation}\label{ea3}
  \sqrt{\frac{T}{\sigma_f^2}}\bw'(\hSig_f-\Sig)\bw\rightarrow^d \mathcal{N}(0,1).
 \end{equation}
\end{lem}
\begin{proof} In fact,
 $$
 \sqrt{\frac{T}{\sigma_f^2}}\bw'(\hSig_f-\Sig)\bw= \sqrt{\frac{T}{\sigma_f^2}}d_1+ \sqrt{\frac{T}{\sigma_f^2}}(d_2+d_3+d_4).
 $$
 By Lemma \ref{la.8}, it suffices to show that $\sqrt{T/\sigma_f^2}(d_2+d_3+d_4)=o_p(1)$. By Lemma \ref{la.4}, $\sqrt{T/\sigma_f^2}|d_2|=O_p((\bw'\Sig_u\bw)^{1/4}(E|\bw'\bu_t|^4)^{1/8}) /\sqrt{\sigma_f^2} )=O_p((\bw'\Sig_u\bw)^{1/4}  /\sqrt{\sigma_f^2} )=o_p(1)$ since $E|\bw'\bu_t|^4=O(1)$. Moreover, Lemma \ref{la.6} implies $\sqrt{T/\sigma_f^2}|d_3|=O_p((\bw'\Sig_u\bw)^{1/2}(T\sigma_f^2)^{-1/2})=o_p(1)$ since $\bw'\Sig_u\bw=o(\sigma_f^4)=O(1)$. It also follows from Lemma \ref{la.6} that  $\sqrt{T/\sigma_f^2}|d_4|=O_p(\bw'\Sig\bw s_N (\sigma_f^2)^{-1/2}(\log N)^{1/2-q/2}T^{q/2})=o_p(1)$. This implies the desired result.
 \end{proof}

\textbf{Proof of Theorem \ref{t32}}

The first statement
 $
\left[\var\left(\sum_{t=1}^T(\bw'\bB\bff_t)^2\right)\right]^{-1/2}   T\bw'(\hSig_f-\Sig)\bw\rightarrow^d\mathcal{N}(0,1)
$
follows from (\ref{ea2}).    
By  the assumptions,  $L\sqrt{(L+\log N)/T}=o(  \sigma^2_f   )$ and $\sum_{h>L}\gamma_f(h)=o(  \sigma^2_f   )$ and Lemma \ref{l32}  imply $|  \sigma^2_f   -\hsig_f^2|=o_p(  \sigma^2_f   )$. Since $\bw'(\hSig_f-\Sig)\bw=O_p(T^{-1/2}\sqrt{  \sigma^2_f   })$, we have
$$
\sqrt{T}|\bw'(\hSig_f-\Sig)\bw|\left|\frac{1}{\sqrt{  \sigma^2_f  }}-\frac{1}{\sqrt{\hsig_f^2}}\right|=o_p(1).
$$
 Hence  Lemma \ref{lb.7} gives
$\sqrt{T/\hsig_f^2}\bw'(\hSig_f-\Sig)\bw\rightarrow^d N(0, 1),$
which validates the H-CLUB.

\section{Proofs for the POET-based Estimation}

   Let $\bV$ denote the $r\times r$ diagonal matrix of the first $r$ largest eigenvalues of $\bS$ in decreasing order. Let $\hF=(\hf_1,...,\hf_T)$ be an $r\times T$ matrix such that the rows of $\hF/\sqrt{T}$ are the eigenvectors corresponding to the $r$ largest eigenvalues of the $T\times T$ matrix $\bR'\bR$. Let $\hB=\bR\hF'/T.$ Define an $r\times r$ matrix
$$
\bH=\frac{1}{T}\bV^{-1}\hF\bF'\bB'\bB.
$$
Then $\hB$ and $\hf_t$ can be treated as estimators of $\bB\bH^{-1}$ and $\bH\bff_t$ respectively.

\subsection{Proof of Lemma \ref{l33}}

\begin{lem}\label{la.10}
(i) $\|\bw'\hB\|=O_p(1)$, and $\|\bw'(\hB-\bB\bH^{-1})\|=O_p(N^{-1/2}+(\log N/T)^{1/2})$\\
(ii) $\|\hF-\bH\bF\|^2/T=T^{-1}\sum_{t=1}^T\|\hf_t-\bH\bff_t\|^2=O_p(N^{-1}+T^{-2})$.\\
(iii) $\|\bw'\bE\|^2=O_p(T).  $\\
(iv) $\|T^{-1}\sum_{t=1}^T[\hf_t\hf_t'-\bH\bff_t(\bH\bff_t)']\|=O_p(N^{-1/2}+T^{-1})$.
\end{lem}

\begin{proof}(i) By Lemma B.16 in an earlier version of Fan et al. (2013)\footnote{downloadable from http://terpconnect.umd.edu/$\sim$yuanliao/factor2/factor2.html }, $\|\hB\|_{\max}\leq \|\hB-\bB\bH^{-1}\|_{\max}+\|\bB\|_{\max}=O_p(1)$. Thus $\|\bw'\hB\|^2\leq r\|\hB\|_{\max}^2\|\bw\|_1^2=O_p(1)$. On the other hand, $\|\bw'(\hB-\bB\bH^{-1})\|^2\leq r\|\hB-\bB\|_{\max}^2\|\bw\|_1^2=O_p(1/N+\log N/T)$.

(ii)  By (A.1) in Bai (2003), the following identity holds:
\begin{equation}\label{eb1}
\hf_t-\bH\bff_t=(\bV/N)^{-1}\left(\frac{1}{T}\sum_{s=1}^T\hf_sE({\bu}_s'{\bu}_t)/N+\frac{1}{T}\sum_{s=1}^T\hf_s\zeta_{st}+\frac{1}{T}\sum_{s=1}^T\hf_s\eta_{st}+\frac{1}{T}\sum_{s=1}^T\hf_s\xi_{st}\right)
\end{equation}
where $\zeta_{st}={\bu}_s'{\bu}_t/N-E({\bu}_s'{\bu}_t)/N$, $\eta_{st}=\bff_s'\sum_{i=1}^N\bb_iu_{it}/N$, and
$\xi_{st}=\bff_t'\sum_{i=1}^p\bb_iu_{is}/N$. It follows from Lemma C.7 in Fan et al. (2013) that
$$
\frac{1}{T}\sum_{t=1}^T(\frac{1}{T}\sum_{s=1}^T\hat{f}_{is} \zeta_{st})^2+\frac{1}{T}\sum_{t=1}^T(\frac{1}{T}\sum_{s=1}^T\hat{f}_{is} \eta_{st})^2+\frac{1}{T}\sum_{t=1}^T(\frac{1}{T}\sum_{s=1}^T\hat{f}_{is} \xi_{st})^2=O_p(\frac{1}{N}).
$$Moreover, by Lemma C.9 of Fan et al. (2013), $\max_{i\leq r}\frac{1}{T}\sum_{t=1}^T(\hf_t-\bH\bff_t)_i^2=O_p(1/T+1/N)$. Applying the inequality $(a+b)^2\leq 2a^2+2b^2$ gives,
\begin{eqnarray*}
&& \frac{1}{T}\sum_{t=1}^T(\frac{1}{T}\sum_{s=1}^T\hat{f}_{is} E(\bu_s'\bu_t)/N)^2
\leq\frac{1}{T}\sum_{t=1}^T(\frac{1}{T}\sum_{s=1}^T[|(\hf_s-\bH\bff_s)_i|+|(\bH\bff_s)_i|]|E(\bu_s'\bu_t)|/N)^2\cr
&\leq&\frac{2}{T}\sum_{t=1}^T(\frac{1}{T}\sum_{s=1}^T|(\hf_s-\bH\bff_s)_i||E(\bu_s'\bu_t)|/N)^2+\frac{2}{T}\sum_{t=1}^T(\frac{1}{T}\sum_{s=1}^T|(\bH\bff_s)_i||E(\bu_s'\bu_t)|/N)^2.
\end{eqnarray*}
By the Cauchy-Schwarz inequality and that $\max_{t\leq T}\sum_{s=1}^T|E(\bu_s'\bu_t)/N|^2=O(1)$,
$$\frac{2}{T}\sum_{t=1}^T(\frac{1}{T}\sum_{s=1}^T|(\hf_s-\bH\bff_s)_i||E(\bu_s'\bu_t)|/N)^2\leq \max_{i\leq r}\frac{2}{T}\sum_{s=1}^T(\hf_s-\bH\bff_s)_i^2\frac{1}{T}\sum_{s=1}^T(|E\bu_s'\bu_t|/N)^2=O_p(\frac{1}{T^2}+\frac{1}{NT}).
$$
Also, $\frac{2}{T}\sum_{t=1}^T(\frac{1}{T}\sum_{s=1}^T|(\bH\bff_s)_i||E(\bu_s'\bu_t)|/N)^2\leq O_p(T^{-1})\sum_{t=1}^T(\frac{1}{T}\sum_{s=1}^T\|\bff_s\||E(\bu_s'\bu_t)|/N)^2.$
We have $T^{-1}\sum_{t=1}^T(\frac{1}{T}\sum_{s=1}^T\|\bff_s\||E(\bu_s'\bu_t)|/N)^2=O_p(T^{-2})$ since
\begin{eqnarray*}
&&E\frac{1}{T}\sum_{t=1}^T(\frac{1}{T}\sum_{s=1}^T\|\bff_s\||E(\bu_s'\bu_t)|/N)^2=\frac{1}{T^2}\sum_{s=1}^T\sum_{l=1}^TE\|\bff_s\|\|\bff_l\|\frac{|E\bu_s'\bu_t|}{N}\frac{|E\bu_l'\bu_t|}{N}\cr
&&\leq \max_{s\leq T}E\|\bff_s\|^2\max_{s\leq T}(\frac{1}{T}\sum_{s=1}^T|E\bu_s'\bu_t|/N)^2=O(T^{-2}).
\end{eqnarray*}
This implies $\max_{i\leq r}T^{-1}\sum_{t=1}^T(\hf_t-\bH\bff_t)_i^2=O_p(N^{-1}+T^{-2})$, and \\thus $T^{-1}\sum_{t=1}^T\|\hf_t-\bH\bff_t\|^2=O_p(N^{-1}+T^{-2})$.\\
(iii) $E\|\bw'\bE\|^2= E\sum_{t=1}^T(\sum_{i=1}^Nw_iu_{it})^2=T\max_{i,j}|Eu_{it}u_{jt}|\|\bw\|_1^2=O(T) 	$. Thus \\ $\|\bw'\bE\|^2=O_p(T)$. Finally, (iv)  follows from the Cauchy-Schwarz inequality and part (ii).

\end{proof}

 \begin{lem}\label{lb4}
$\max_{h\leq L}|\hgamma_P(h)-\gamma_f(h)|=O_p(\sqrt{(L+\log N)/T}+N^{-1/2}).$
\end{lem}
\begin{proof}

The triangular inequality implies $\max_{h\leq L}|\hgamma_P(h)-\gamma_f(h)|\leq\sum_{i=1}^4a_i$, where
\begin{eqnarray*}
&a_1=\max_{h\leq L}|T^{-1}\sum_{t=1}^{T-h}(\bw'\hB\hf_{t+h})^2(\bw'\hB\hf_t)^2-E(\bw'\bB\bff_t)^2(\bw'\bB\bff_{t+h})^2|, \cr
&a_2= |(\bw'\hB\hB'\bw)^2-(\bw'\bB\bB'\bw)^2|,\quad
a_3=\bw'\hB\hB'\bw\max_{h\leq L}|\bw'\hB\hB'\bw-T^{-1}\sum_{t=1}^{T-h}(\bw'\hB\hf_t)^2|, \cr
&a_4=\bw'\hB\hB'\bw\max_{h\leq L}|\bw'\hB\hB'\bw-T^{-1}\sum_{t=1}^{T-h}(\bw'\hB\hf_{t+h})^2|.
\end{eqnarray*}
Here $a_1$ is bounded by $a_{11}+a_{12}$, where \\
 $a_{11}=
 \max_{h\leq L}|T^{-1}\sum_{t=1}^{T-h}(\bw'\bB\bff_{t+h})^2(\bw'\bB\bff_t)^2-E(\bw'\bB\bff_t)^2(\bw'\bB\bff_{t+h})^2|,
 $
and \\$a_{12}=\max_{h\leq L}|T^{-1}\sum_{t=1}^{T-h}(\bw'\hB\hf_{t+h})^2(\bw'\hB\hf_t)^2-(\bw'\bB\bff_{t+h})^2(\bw'\bB\bff_t)^2|$.

  As in the proof of Lemma \ref{la.9}, $a_{11}=O_p(\sqrt{L/T}).$   It follows from the Cauchy-Schwarz inequality that
\begin{eqnarray*}
a_{12}&=&O_p(\max_{h\leq L}|T^{-1}\sum_{t=1}^{T-h}(\bw'\hB\hf_{t+h})(\bw'\hB\hf_t)-(\bw'\bB\bff_{t+h})(\bw'\bB\bff_t)|)\cr
&=&O_p(\max_{h\leq L}(T^{-1}\sum_{t=1}^{T-h}(\bw'\hB\hf_{t+h}-\bw'\bB\bff_{t+h})^2)^{1/2}+(T^{-1}\sum_{t=1}^{T-h}(\bw'\hB\hf_{t}-\bw'\bB\bff_{t})^2)^{1/2})\cr
&=&O_p(\|\bw'(\hB-\bB\bH^{-1})\|+\|\bw'\bB\bH^{-1}\|(T^{-1}\sum_{t=1}^T\|\hf_t-\bH\bff_t\|^2)^{1/2})
\end{eqnarray*}
It follows from  Lemma \ref{la.10} that $a_{12}=O_p(\sqrt{\log N/T}+N^{-1/2})$, thus\\
 $a_1=O_p(\sqrt{(L+\log N)/T}+N^{-1/2})$. On the other hand, for $g_1,...,g_5$ defined in (\ref{eqa.1}),
 \begin{eqnarray*}
 a_2=O_p(|\bw'(\hB\hB'-\bB\bB')\bw|)=O_p(\sum_{i=1}^5|g_i|)=O_p(\sqrt{\sigma_f^2/T}).
 \cr
 a_3=O_p(1)\max_{h\leq L}|\bw'\hB(\frac{1}{T}\sum_{t=T-h+1}^T\hf_t\hf_t')\hB'\bw|=O_p(\frac{1}{T}\sum_{t=T-L+1}^T\|\hf_t\|^2).
 \end{eqnarray*}
We have $\frac{1}{T}\sum_{t=T-L+1}^T\|\hf_t\|^2\leq \frac{2}{T}\sum_{t=T-L+1}^T\|\hf_t-\bH\bff_t\|^2+\frac{2}{T}\sum_{t=T-L+1}^T\|\bH\bff_t\|^2$. On one hand, $E\frac{2}{T}\sum_{t=T-L+1}^T\|\bff_t\|^2=O(L/T).$ On the other hand, by Theorem 3.3 in Fan et al. (2013), $ \max_{t}\|\hf_t-\bH\bff_t\|=o_p(1)$, hence $ \frac{2}{T}\sum_{t=T-L+1}^T\|\hf_t-\bH\bff_t\|^2=O_p(L/T)$. Thus $a_3=O_p(L/T)$. Similarly, we have $a_4=O_p(L/T)$.

\end{proof}

\textbf{Proof of Lemma \ref{l33}}

We have  $|\hsig_P^2-  \sigma^2_f   |\leq\sum_{i=1}^3b_i$, where
$
b_1=|\hgamma_P(0)-\gamma_f(0)|, $ $ b_2=2\sum_{h=1}^L|\hgamma_P(h)-\gamma_f(h)|, $ and $ b_3=2\sum_{h>L}\gamma_f(h).
$
 Lemma \ref{lb4} implies $b_2\leq 2L\max_{h\leq L}|\hgamma_P(h)-\gamma_f(h)|=O_p(L\sqrt{(L+\log N)/T}+L/\sqrt{N})$, which gives the convergence rate.  The second statement  is due to $\hsig_f^2=O_p( \log N)$.

\subsection{Proof of Theorem \ref{t33}}

First,  $\hSig_P=\hB\hB'+\bOmega$. With the identification condition $\cov(\bff_t)=\bI_K$, $\Sig=\bB\bB'+\Sig_u.$ Therefore, if we write $\tC_T=\hB-\bB\bH^{-1}$ and $\tD_T=T^{-1}\sum_{t=1}^T\bff_t\bff_t'-\cov(\bff_t)$, then $\bw'(\hSig_P-\Sig)\bw=\sum_{i=1}^5g_i$, where
\begin{eqnarray}\label{eqa.1}
g_1=\bw'\bB\tD_T\bB'\bw, g_2=\bw'\tC_T\hB'\bw, g_3=\bw'\bB\bH^{-1}\tC_T'\bw, \cr
g_4=\bw'(\bOmega-\Sig_u)\bw,  g_5=\bw'\bB\bH^{-1}\frac{1}{T}\sum_{t=1}^T[\hf_t\hf_t'-\bH\bff_t(\bH\bff_t)']\bH^{-1'}\bB'\bw.
\end{eqnarray}
Recall the definition of $d_1$ in Appendix \ref{sb2},  $g_1=d_1$. Thus  it follows from Lemma \ref{la.8} that $\sqrt{T/\sigma_f^2}g_1\rightarrow^d \mathcal{N}(0,1)$. We proceed by showing that $\sqrt{T/\sigma_f^2}g_i$ are asymptotically negligible for $i=2,...,5$. These results are given in the following lemmas.

\begin{lem}\label{la.11}
(i) $|g_2|=O_p(T^{-1/2}(\bw'\Sig_u\bw)^{1/4}+N^{-1/2}+T^{-1})$,  \\
(ii)$|g_3|=O_p(T^{-1/2}(\bw'\Sig_u\bw)^{1/4}+N^{-1/2}+T^{-1})$.\\
(iii) $|g_5|=O_p(N^{-1/2}+T^{-1}).$
\end{lem}
\begin{proof}

Using the facts that $\bR=\bB\bF+\bE$, $\hB=\bR\hF'/T$ and $\hF\hF'/T=\bI_K$, we have
$$
\hB-\bB\bH^{-1}=\bB\bH^{-1}(\bH\bF-\hF)\hF'/T+\bE(\hF-\bH\bF)'/T+\bE\bF'\bH'/T.
$$
Thus $g_2=g_{21}+g_{22}+g_{23}$, where $g_{21}=\bw' \bB\bH^{-1}(\bH\bF-\hF)\hF'/T  \hB'\bw$, \\ $g_{22}=\bw'\bE(\hF-\bH\bF)'/T\hB'\bw$, and $g_{23}=\bw'\bE\bF'\bH'/T\hB'\bw$.   It was shown by Fan et al. (2013) that $\|\bH\|=O_p(1)=\|\bH^{-1}\|$. Thus by Lemma  \ref{la.10}, $|g_{21}|\leq O_p(1)\|\bH\bF-\hF\|\|\hF\|/T=O_p(\sqrt{1/N+1/T^2}).$ In addition, $|g_{22}|\leq O_p(\sqrt{T})\|\hF-\bH\bF\|/T=O_p(\sqrt{1/N+1/T^2})$. Finally,  by Lemma \ref{la.4}$, |g_{23}|\leq \|\bw'\bE\bF'\|O_p(1/T)= O_p(T^{-1/2}(\bw'\Sig_u\bw)^{1/4})$. The proof for the convergence rate of $|g_3|$ is the same, so is omitted. Finally, the rate of convergence for $|g_5|$ follows from  Lemma \ref{la.10}.
\end{proof}

\textbf{Proof of Theorem \ref{t33}}

 By Lemma \ref{la.11} and the assumption that $T\sigma_f^2\rightarrow\infty$,  $\sigma_f^2N/T\rightarrow\infty$ and $\bw'\Sig_u\bw=o(\sigma_f^4)$,  we have  $\sqrt{T/\sigma_f^2}g_i=o_p(1)$ for $i=2, 3, 5$. In addition, by Theorem 3.1 of Fan et al. (2013), $$\|\bOmega-\Sig_u\|=O_p(s_N(\frac{\log N}{T}+\frac{1}{N})^{1/2-q/2}),$$
which then implies that $\sqrt{T/\sigma_f^2}g_4=o_p(1)$, by the assumption \\ $\|\bw\|^2=o(\sqrt{\sigma_f^2}s_N^{-1}N^{1/2-q/2}T^{-1/2}).$ By Lemma \ref{la.8}, $\sqrt{T/\sigma_f^2}g_1=\sqrt{T/\sigma_f^2}d_1\rightarrow^d \mathcal{N}(0,1)$,  thus
$\sqrt{T/\sigma_f^2}\bw'(\hSig_P-\Sig)\bw\rightarrow^d \mathcal{N}(0,1).$  By the assumption  $L=O(\sqrt{N}\sigma_f^2)$ and Lemma \ref{l33},   $|  \sigma^2_f   -\hsig_P^2|=o_p(  \sigma^2_f   )$.  Since $\bw'(\hSig_P-\Sig)\bw=O_p(T^{-1/2}\sqrt{  \sigma^2_f   })$,  we have
 $\sqrt{T}|\bw'(\hSig_P-\Sig)\bw||\sigma_f^{-1}-\hsig_P^{-1}|=o_p(1)$.
 It follows  that  $\sqrt{T/\hsig_P^2}\bw'(\hSig_P-\Sig)\bw\rightarrow^d N(0, 1)
$, which validates the H-CLUB.

\subsection{Proof of Theorem \ref{t44}}
The theorem follows from the following lemma.
\begin{lem} Suppose $\bff_t$ and $\bu_t$ are independent, and $E\bff_t=E\bu_t=0$, then
 $$\var\left[\sum_{t=1}^T(\bw'\bR_t)^2\right]=\var\left[\sum_{t=1}^T(\bw'\bB\bff_t)^2\right]+\var\left[\sum_{t=1}^T(\bw'\bu_t)^2\right]+\var\left[2\sum_{t=1}^T\bw'\bB\bff_t\bw'\bu_t	 \right].$$
 \end{lem}
\proof Since $\bR_t=\bB\bff_t+\bu_t$, and $\bff_t$ and $\bu_t$ are independent,
\begin{eqnarray*}
&&\var\left[\sum_{t=1}^T(\bw'\bR_t)^2\right]=\var\left[\sum_{t=1}^T(\bw'\bB\bff_t)^2+(\bw'\bu_t)^2+2\bw'\bB\bff_t\bw'\bu_t		 \right]\cr
&&=\var\left[\sum_{t=1}^T(\bw'\bB\bff_t)^2\right]+\var\left[\sum_{t=1}^T(\bw'\bu_t)^2\right]+\var\left[2\sum_{t=1}^T\bw'\bB\bff_t\bw'\bu_t	 \right]\cr
&&+2\cov\left[		\sum_{t=1}^T(\bw'\bB\bff_t)^2+(\bw'\bu_t)^2		 ,2\sum_{t=1}^T\bw'\bB\bff_t\bw'\bu_t		 \right].
\end{eqnarray*}
It suffices to show the covariance term is zero. In fact,  since $E\bw'\bB\bff_t\bw'\bu_t=0,$
\begin{eqnarray*}
&&\cov\left[		\sum_{t=1}^T(\bw'\bB\bff_t)^2 	,\sum_{t=1}^T\bw'\bB\bff_t\bw'\bu_t		 \right]=\sum_{s\leq T,t\leq T}\cov[(\bw'\bB\bff_s)^2, \bw'\bB\bff_t\bw'\bu_t]\cr
&&=\sum_{s, t}E(\bw'\bB\bff_s)^2\bw'\bB\bff_t\bw'\bu_t=\sum_{s, t}E(\bw'\bB\bff_s)^2\bw'\bB\bff_tE\bw'\bu_t=0.
\end{eqnarray*}
Finally, as  $E\bff_t=0$ implies $E\bw'\bB\bff_t=0$, we have
\begin{eqnarray*}
&&\cov\left[		\sum_{t=1}^T(\bw'\bu_t)^2 	,\sum_{t=1}^T\bw'\bB\bff_t\bw'\bu_t		 \right]=\sum_{s\leq T,t\leq T}\cov[(\bw'\bu_s)^2, \bw'\bB\bff_t\bw'\bu_t]\cr
&&=\sum_{s, t}E(\bw'\bu_s)^2\bw'\bB\bff_t\bw'\bu_t=\sum_{s, t}E(\bw'\bB\bff_t)E(\bw'\bu_s)^2\bw'\bu_t=0.
\end{eqnarray*}

\end{document}